\documentclass[11pt,reqno]{article}
\pdfoutput=1 
\usepackage[utf8]{inputenc}
\usepackage[T1]{fontenc}
\usepackage{libertine}
\usepackage{libertinust1math}
\usepackage{microtype}
\usepackage{algorithm}
\usepackage[noend]{algpseudocode}
\usepackage{amsmath, amssymb, amsthm}
\usepackage{bm}
\usepackage{graphicx}
\usepackage{thmtools}
\usepackage{thm-restate}
\usepackage{hyperref} 
\usepackage[capitalize, nameinlink]{cleveref}
\usepackage{comment}
\usepackage{enumerate}
\usepackage{fullpage}
\usepackage{mathtools}
\usepackage{subcaption}
\usepackage[dvipsnames]{xcolor}
\usepackage{nag}
\usepackage{authblk}

\graphicspath{{figures/}{.}}

\newcommand{\declarecolor}[2]{\definecolor{#1}{RGB}{#2}\expandafter\newcommand\csname #1\endcsname[1]{\textcolor{#1}{##1}}}
\declarecolor{White}{255, 255, 255}
\declarecolor{Black}{0, 0, 0}
\declarecolor{LightGray}{216, 216, 216}
\declarecolor{Gray}{127, 127, 127}
\declarecolor{Orange}{237, 125, 49}
\declarecolor{LightOrange}{251,229, 214}
\declarecolor{Yellow}{255, 192, 0}
\declarecolor{LightYellow}{255, 242, 200}
\declarecolor{Blue}{91, 155, 213}
\declarecolor{LightBlue}{222, 235, 247}
\declarecolor{Green}{112, 173, 71}
\declarecolor{LightGreen}{226, 240, 217}
\declarecolor{Navy}{68, 114, 196}
\declarecolor{LightNavy}{218, 227, 243}

\hypersetup{
	colorlinks=true,
	pdfpagemode=UseNone,
	citecolor=Green,
	linkcolor=Navy,
	urlcolor=Navy,
	pdfstartview=FitW,
}

\crefformat{equation}{(#2#1#3)}

\newcommand{\tv}[1]{\left\|#1\right\|_{\normalfont\text{TV}}}
\newcommand{\poly}{\normalfont\text{poly}}
\newcommand{\Prob}[1]{\Pr\parens*{#1}}
\newcommand{\ProbCond}[2]{\Pr\parens*{#1 \;\middle|\; #2}}

\newcommand{\DEF}{\stackrel{\textnormal{\tiny\sffamily def}}{=}}
\newcommand{\Z}{\mathbb{Z}}
\newcommand{\R}{\mathbb{R}}

\newcommand{\CR}{C_{\textnormal{R}}}
\newcommand{\CFL}{C_{\textnormal{FL}}}
\newcommand{\CAFL}{C_{\textnormal{AFL}}}
\newcommand{\CG}{C_{\textnormal{G}}}
\newcommand{\xR}{x_{\textnormal{R}}}
\newcommand{\xG}{x_{\textnormal{G}}}
\newcommand{\OMEGALEFT}{\Omega_{\textnormal{LEFT}}}
\newcommand{\OMEGAMIDDLE}{\Omega_{\textnormal{MIDDLE}}}
\newcommand{\OMEGARIGHT}{\Omega_{\textnormal{RIGHT}}}

\newcommand{\cM}{\mathcal{M}}

\DeclarePairedDelimiter{\abs}{\lvert}{\rvert}
\DeclarePairedDelimiter{\set}{\{}{\}}
\DeclarePairedDelimiter{\parens}{(}{)}
\DeclarePairedDelimiter{\bracks}{[}{]}
\DeclarePairedDelimiter{\floor}{\lfloor}{\rfloor}

\DeclarePairedDelimiter{\norm}{\lVert}{\rVert}

\theoremstyle{plain}
\newtheorem{theorem}{Theorem}[section]
\newtheorem{lemma}[theorem]{Lemma}

\theoremstyle{definition}

\newtheorem{remark}[theorem]{Remark}

\makeatletter

\newcommand*\wthelper[2]{%
        \hbox{\dimen@\accentfontxheight#1%
                \accentfontxheight#11.1\dimen@
                $\m@th#1\widetilde{#2}$%
                \accentfontxheight#1\dimen@
        }%
}
\newcommand*\accentfontxheight[1]{%
        \fontdimen5\ifx#1\displaystyle
                \textfont
        \else\ifx#1\textstyle
                \textfont
        \else\ifx#1\scriptstyle
                \scriptfont
        \else
                \scriptscriptfont
        \fi\fi\fi3
}
\makeatother

\begin{document}

\title{Slow Mixing of Glauber Dynamics for the Six-Vertex Model in the Ordered
Phases}

\author{Matthew Fahrbach\thanks{
Email: \href{mailto:matthew.fahrbach@gatech.edu}{\textsf{matthew.fahrbach@gatech.edu}}.
Supported in part by an NSF Graduate Research Fellowship under grant DGE-1650044.
}\hspace{0.05cm}}
\author{Dana Randall\thanks{
Email: \href{mailto:randall@cc.gatech.edu}{\textsf{randall@cc.gatech.edu}}.
Supported in part by NSF grants CCF-1526900, CCF-1637031, and CCF-1733812.
}}
\affil{School of Computer Science, Georgia Institute of Technology}

\maketitle

\begin{abstract}
The six-vertex model in statistical physics is a weighted generalization of
the ice model on $\Z^2$ (i.e., Eulerian orientations) and
the zero-temperature three-state Potts model  (i.e., proper three-colorings).
The phase diagram of the model depicts its physical properties
and suggests where local Markov chains will be efficient.
In this paper, we analyze the mixing time of Glauber dynamics for the six-vertex
model in the ordered phases.  Specifically, we show that for
all Boltzmann weights in the \emph{ferroelectric phase}, there exist boundary
conditions such that local Markov chains require exponential time to converge
to equilibrium.  This is the first rigorous result bounding the mixing time of
Glauber dynamics in the ferroelectric phase.
Our analysis demonstrates a fundamental connection between
correlated random walks and the dynamics of intersecting lattice path models (or routings).
We analyze the Glauber dynamics for the six-vertex model with
free boundary conditions in the \emph{antiferroelectric phase} and significantly
extend the region for which local Markov chains are known to be slow mixing.
This result relies on a Peierls argument and novel properties of weighted
non-backtracking walks.
\end{abstract}

\pagenumbering{gobble}
\clearpage
\pagenumbering{arabic}

\section{Introduction}
\label{sec:introduction}

The \emph{six-vertex model} was first introduced by Pauling in
1935~\cite{pauling} to study the thermodynamics of crystalline solids with
ferroelectric properties, and has since become one of the most compelling
models in statistical mechanics.
The prototypical instance of the model is the hydrogen-bonding
pattern of two-dimensional ice---when water freezes, each oxygen atom must
be surrounded by four hydrogen atoms such that two of the hydrogen atoms
bond covalently with the oxygen atom and two are farther away.  The state space of
the six-vertex model consists of orientations of the edges in a
finite region of the two-dimensional square lattice where every internal vertex
has two incoming edges and two outgoing edges, also represented as Eulerian
orientations of the underlying lattice graph.
The model is most often studied on the $n \times n$ square
lattice~$\Lambda_n \subseteq \Z^2$ with~$4n$ additional edges so that each
internal vertex has degree~4.
There are six possible edge orientations incident to a vertex (see \Cref{fig:six}).
We assign Boltzmann weights $w_1,w_2,w_3,w_4,w_5,w_6 \in \R_{> 0}$ to the
six vertex types and define
the partition function as
$Z = \sum_{x \in \Omega} \prod_{i=1}^6 w_i^{n_i(x)}$,
where~$\Omega$ is the set of Eulerian orientations of~$\Lambda_n$
and $n_i(x)$ is the number of type-$i$ vertices in the configuration $x$.

\begin{figure}[H]
  \vspace{0.45cm}
  \centering
  \includegraphics[width=0.85\linewidth]{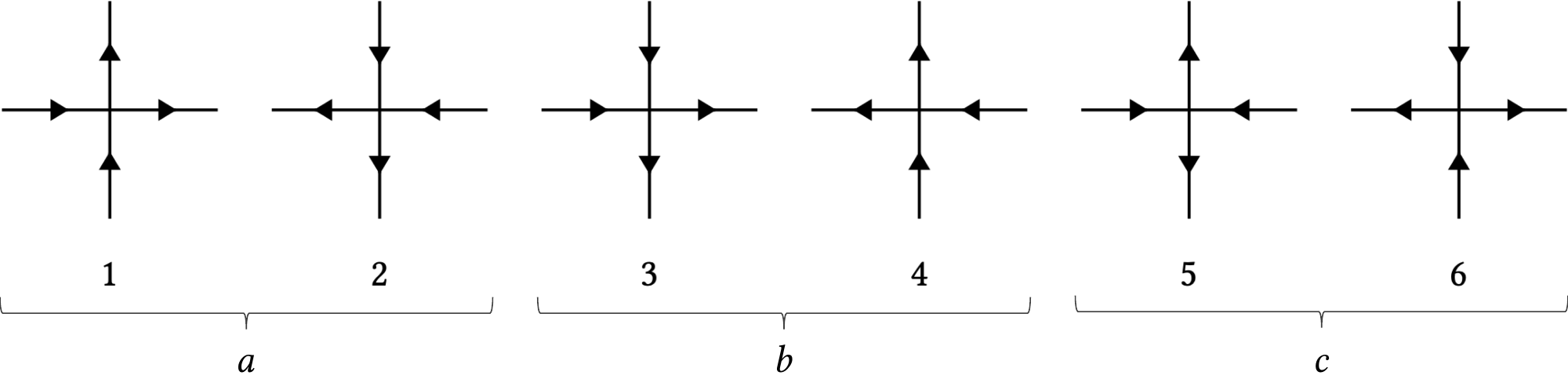}
  \vspace{0.10cm}
  \caption{The six valid edge orientations for internal vertices in the six-vertex model
  and their corresponding Boltzmann weights.}
  \label{fig:six}
\end{figure}

In 1967, Lieb discovered exact solutions to the six-vertex model with
periodic boundary conditions (i.e., on the torus) for three different parameter
regimes~\cite{lieb1967exact,lieb1967kdp,lieb1967residual}.
In particular, he famously showed that if all six vertex weights are set to $w_i = 1$,
the energy per
vertex is $\lim_{n\rightarrow\infty} Z^{1/n^2} = \parens*{4/3}^{3/2} = 1.5396007...$,
which is known as ``Lieb's square ice constant''.
His results were immediately generalized to all parameter regimes
and to account for external electric fields~\cite{sutherland1967exact,yang1967exact}.
An equivalence between periodic and free boundary conditions in the limit
was established in~\cite{brascamp1973some}, and since then the primary
object of study has been the six-vertex model subject to
\emph{domain wall boundary conditions},
where the lower and upper boundary edges point into the square
and the left and right boundary edges point outwards~\cite{izergin1992determinant,korepin2000thermodynamic,bogoliubov2002boundary,bleher2006exact,bleher2009exact,bleher2010exact}.
The six-vertex model serves as an important ``counterexample''
in statistical physics because the \emph{surface free energy} in the thermodynamic
limit depends on the boundary conditions. In particular, it is different
for periodic and domain wall boundary conditions.

There have been several surprisingly profound connections
to combinatorics and probability in this line of work.
For example, Zeilberger gave a sophisticated computer-assisted proof of the
\emph{alternating sign matrix conjecture} in 1995~\cite{zeilberger1996proof}.
A year later,
Kuperberg~\cite{kuperberg1996another} produced an elegant and significantly
shorter proof using analysis of the partition function of the six-vertex model
with domain wall boundary conditions.
Other connections to combinatorics
include the dimer model on the Aztec diamond and the arctic circle
theorem~\cite{cohn1996,ferrari2006domino},
sampling lozenge tilings~\cite{luby2001markov,wilson2004mixing,bhakta2017approximately},
and counting 3-colorings of lattice graphs~\cite{rt,cannon2016sampling}.

While there has been extraordinary progress in understanding properties of the
six-vertex model with periodic or domain wall boundary conditions
in mathematical physics,
remarkably less is known when the model is subject to arbitrary boundary conditions.
Sampling configurations using Markov chain Monte Carlo (MCMC) algorithms
has been one of the primary means for discovering mathematical and
physical properties of the six-vertex
model~\cite{allison2005numerical,lyberg2017density,lyberg2018phase,keating2018random,belov2020two}.
However, the model is empirically very sensitive to boundary conditions,
and numerical studies have often observed slow convergence of local MCMC algorithms
under certain parameter settings. For example, according to~\cite{lyberg2018phase},
``it must be stressed that the Metropolis algorithm might be impractical
in the antiferromagnetic phase, where the system may be unable to thermalize.''
There are very few rigorous results about natural Markov chains and the
computational complexity of sampling from the six-vertex model
when the Boltzmann distribution is nonuniform,
thus motivating our study of Glauber dynamics for the six-vertex model, the most widely used MCMC
sampling algorithm, in the \emph{ferroelectric} and
\emph{antiferroelectric phases}.

At first glance, the model has six degrees of freedom. However, this
conveniently reduces to a two-parameter family because of invariants
that relate pairs of vertex types.
To see this, it is useful to view the configurations of the six-vertex model as 
intersecting lattice paths by erasing all of the edges that are directed
south or west and keeping the others (see \Cref{fig:example-states}).
Using this bijective ``routing interpretation,'' it is simple to see that 
the number of type-5 and type-6 vertices must be closely correlated.
In addition to revealing invariants, the lattice path
representation of configurations turns out to be exceptionally useful
for analyzing Glauber dynamics.
Moreover, the total weight of a configuration should remain unchanged if all the
edge directions are reversed in the absence of an external electric field,
so we let $w_1 = w_2 = a$, $w_3 = w_4 = b$, and $w_5 = w_6 = c$.
This complementary invariance is known as the \emph{zero field assumption},
and it is often convenient to exploit the conservation laws of the
model~\cite{bleher2009exact} to reparameterize the system so that $w_1 = a^2$
and $w_2 = 1$.
This allows us to ignore empty sites and focus solely on weighted lattice
paths.
Furthermore, since our goal is to sample configurations from the Boltzmann
distribution, we can normalize the partition function by a factor of~$c^{-n^2}$
and consider the weight $(a/c, b/c, 1)$ instead of the parameter $(a,b,c)$.
Collectively, we refer to these properties as the invariance of the Gibbs
measure for the six-vertex model.

\begin{figure*}
\centering
\begin{subfigure}{0.495\textwidth}
  \centering
  \includegraphics[width=0.75\linewidth]{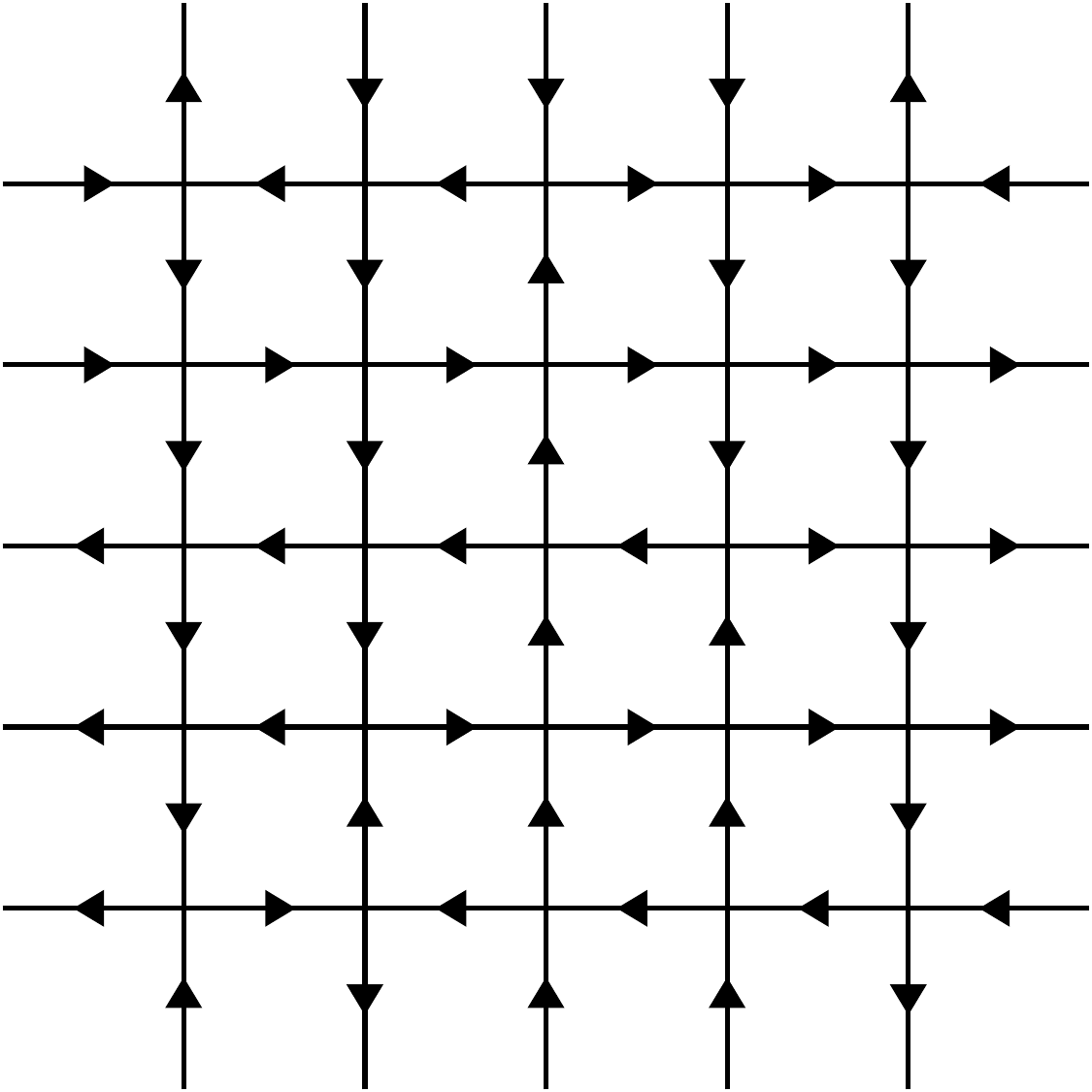}
  \caption{}
  \label{fig:example-state-1}
\end{subfigure}
\hspace{-0.55cm}
\begin{subfigure}{0.495\textwidth}
  \centering
  \includegraphics[width=0.75\linewidth]{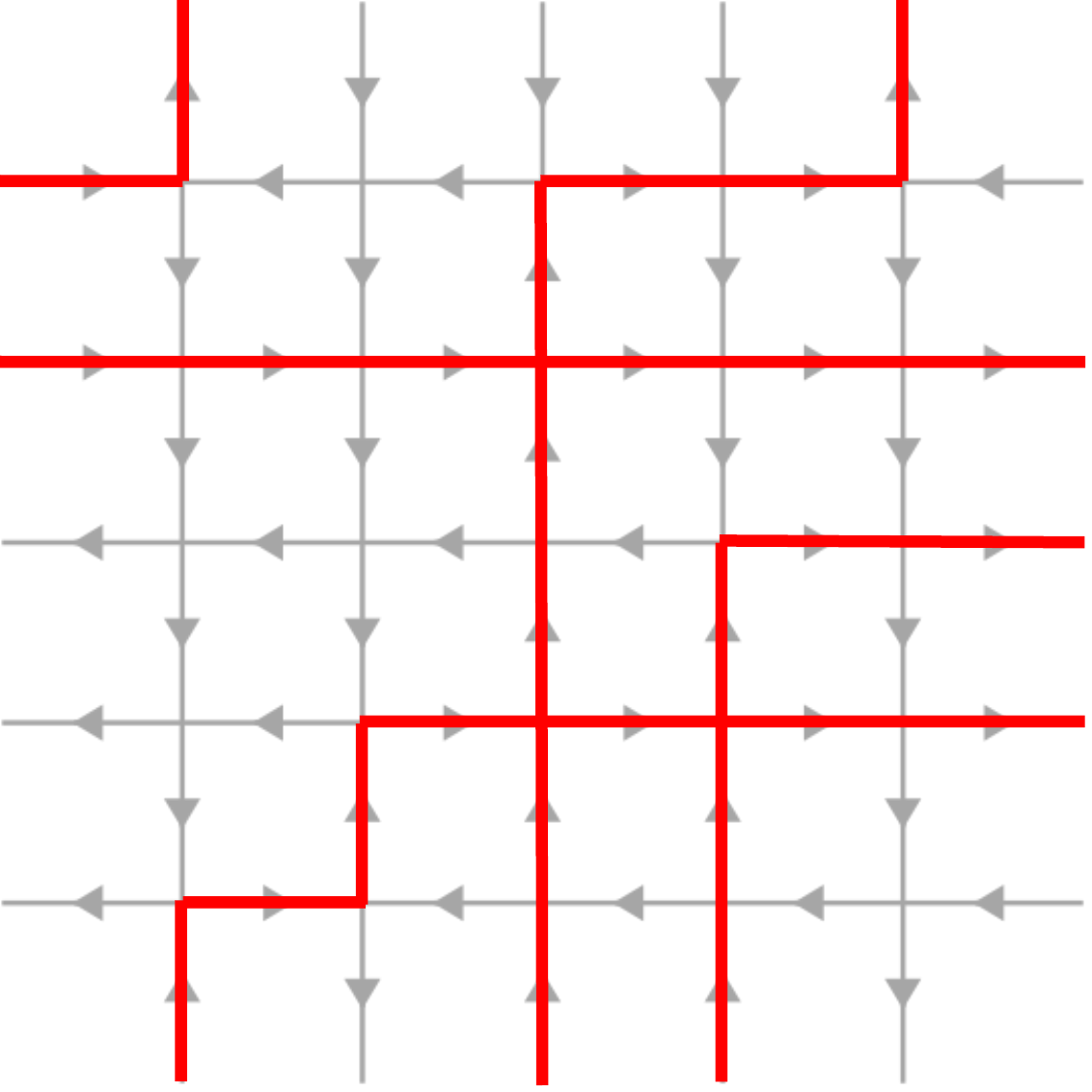}
  \caption{}
  \label{fig:example-state-2}
\end{subfigure}
  \caption{Examples of a configuration in the six-vertex model:
  (a) illustrates the edge orientations and internal Eulerian constraints, and
  (b) overlays the corresponding routing interpretation in red.}
\label{fig:example-states}
\end{figure*}

The single-site \emph{Glauber dynamics} for the six-vertex
model is the Markov chain that makes local moves by (1) choosing an internal cell
of the lattice uniformly at random and (2) reversing the orientations of the
edges that bound the chosen cell if they form a cycle. 
In the lattice path interpretation, these dynamics correspond to
the ``mountain-valley'' Markov chain that flips corners.
Transitions between states are made according to the Metropolis-Hastings
acceptance probability~\cite{metropolis1953equation} so that the Markov chain
converges to the desired distribution.

The phase diagram of the six-vertex model represents distinct
thermodynamic properties of the system and is partitioned into three regions:
the \emph{disordered} (DO) phase,
the \emph{ferroelectric} (FE) phase,
and the \emph{antiferroelectric} (AFE) phase.
To establish these regions, we consider the parameter
\[
  \Delta = \frac{a^2 + b^2 - c^2}{2ab}.
\]
The disordered phase is the set of parameters $(a,b,c) \in \R_{> 0}^3$ that
satisfy $|\Delta| < 1$, and Glauber dynamics is expected to be rapidly mixing
in this region because there are no long-range correlations in the system.
The ferroelectric phase is defined by $\Delta > 1$,
or equivalently when
we have $a > b + c$ or $b > a + c$.
The antiferroelectric phase is defined by $\Delta < -1$, or
equivalently when $a + b < c$.

The phase diagram is symmetric over the positive diagonal, which
follows from the fact that $a$ and $b$ are interchangeable under the automorphism 
that rotates each of the six vertex types by
ninety degrees clockwise.  This is equivalent
to rotating the entire model under the zero field assumption.
Therefore, we can assume
that mixing results are symmetric over the main diagonal.
Combinatorially, we show in~\Cref{sec:ferroelectric} that
configurations in the ferroelectric phase can be interpreted as intersecting lattice paths that
prefer to adhere to each other. We carefully exploit this property to
show that Glauber dynamics slow mixing.
In the antiferroelectric phase, configurations prefer vertices of type-$c$
and tend to be closely aligned with one of states with maximum probability
that are arrow reversals of each other.

\begin{figure*}
\centering
\begin{subfigure}{0.32\textwidth}
  \centering
  \includegraphics[width=1.00\linewidth]{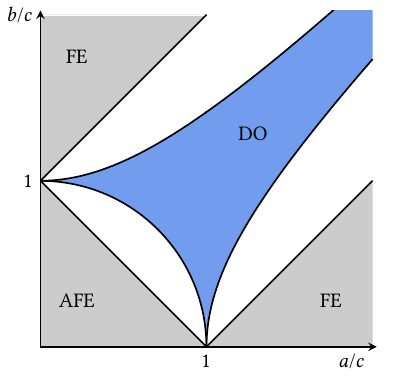}
  \caption{}
  \label{fig:cll-approximability}
\end{subfigure}
\begin{subfigure}{0.32\textwidth}
  \centering
  \includegraphics[width=1.00\linewidth]{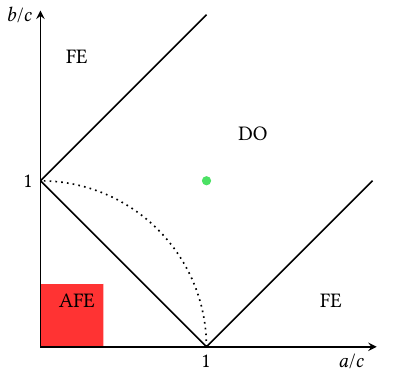}
  \caption{}
  \label{fig:phase-diagram-liu}
\end{subfigure}
\hspace{-0.25cm}
\begin{subfigure}{0.32\textwidth}
  \centering
  \includegraphics[width=1.00\linewidth]{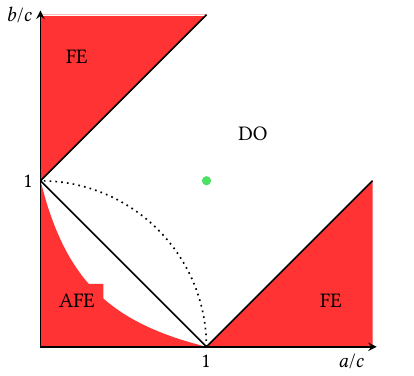}
  \caption{}
  \label{fig:phase-diagram-us}
\end{subfigure}
\caption{
Diagram (a) shows the complexity of approximating the
partition function of the six-vertex model for 4-regular graphs.
There exists an FPRAS in the blue subregion of the disordered phase,
and there cannot exist an FRPAS in the ferroelectric or antiferroelectric (gray)
regions unless $\textbf{RP} = \textbf{NP}$.
Diagram (b) shows the previously known slow mixing regions of Glauber dynamics
in red, and diagram (c) shows the current slow mixing regions.
Glauber dynamics is conjectured to be rapidly mixing in all of the disordered
phase, but it has only been shown for the uniform distribution
indicated by the green point $(1,1)$.
}
\label{fig:phases}
\end{figure*}

\subsection{Related Works}

Cai, Liu, and Lu~\cite{cai2019approximability} recently
investigated the six-vertex model for 4-regular graphs
and provided strong evidence that the complexity of approximating
the partition function agrees with the phase diagram from statistical physics.
In particular, they give a \emph{fully randomized approximation scheme}
(FPRAS)
for all 4-regular graphs in the subregion of the disordered phase
defined by the inequalities $a^2 \le b^2 + c^2$, $b^2 \le a^2 + c^2$, and
$c^2 \le a^2 + b^2$ (i.e., the blue region in~\Cref{fig:cll-approximability}).
Their algorithm builds on the \emph{winding technique} for Holant
problems developed in~\cite{mcquillan2013approximating,huang2016canonical}
and requires $O(n^{10})$ time to sample a six-vertex configuration from the
Boltzmann distribution,
where~$n$ is the number of vertices in the graph.
The Markov chain they use is not Glauber dynamics,
but rather a \emph{directed loop algorithm} whose state space is augmented
with ``near-perfect'' configurations that slightly
violate the Eulerian orientation constraint.
This Markov chain can be understood as gradually reversing a large directed loop in a valid six-vertex
configuration, whereas Glauber dynamics is restricted to reversing
cycles that form the perimeter of a cell.
Cai, Liu, and Lu also showed that an FPRAS for 4-regular graphs cannot exist in the
ferroelectric or antiferroelectric regions unless $\textbf{RP} = \textbf{NP}$
(i.e., the gray regions in~\Cref{fig:cll-approximability}).
Their hardness results use nonplanar 4-regular gadgets to
reduce from 3-MIS, the \textbf{NP}-hard problem of computing the cardinality
of a maximum independent set in a 3-regular graph~\cite{garey1974some},
and therefore so not directly reveal anything about the mixing time of
Glauber dynamics for the six-vertex model on regions of $\Z^2$.
A dichotomy
theorem for the (exact) computability of the partition function of
the six-vertex model on 4-regular graphs was also recently proven
in~\cite{cai2018complexity}.

As for the positive results about the mixing time of Glauber dynamics,
Luby, Randall, and Sinclair~\cite{luby2001markov} proved rapid mixing of a Markov chain
that leads to a fully polynomial almost uniform sampler for
Eulerian orientations on any region of the Cartesian
lattice with fixed boundaries (i.e., the unweighted case when $a/c=b/c=1$).
Randall and Tetali~\cite{rt} then used a comparison technique to argue that
Glauber dynamics for Eulerian orientations on lattice graphs is rapidly
mixing by relating this Markov chain to the Luby-Randall-Sinclair chain.
Goldberg, Martin, and Paterson~\cite{goldberg2004random} extended their
approach to show that Glauber dynamics is rapidly mixing on rectangular lattice
regions with free boundary conditions.

Liu~\cite{liu} gave the first rigorous result showing that Glauber dynamics
can be slowly mixing in a subregion of an ordered phase.
In particular, Liu showed that local Markov chains subject to free boundary
conditions require exponential time to converge to stationarity
in the antiferroelectric subregion
defined by $\max(a,b) < c / \mu$ (i.e., the red region in \Cref{fig:phase-diagram-liu}),
where $\mu = 2.6381585...$ is the \emph{connective constant} for self-avoiding walks on the square lattice.
We note that the connective constant is defined by the limit
$\mu = \lim_{n \rightarrow \infty} \gamma_n^{1/n}$, where $\gamma_n$ is the
number of self-avoiding walks of length $n$ on the square lattice.
Liu also showed that the directed loop algorithm used in~\cite{cai2019approximability}
mixes slowly in the same antiferroelectric subregion
and for all of the ferroelectric region.
This, however, has no bearing on the efficiency
of Glauber dynamics in the ferroelectric region.
As an aside, we also remark that the partition function is exactly
computable for all boundary conditions at the free-fermion point when $\Delta=0$,
or equivalently $a^2 + b^2 = c^2$, via a reduction to domino
tilings and a Pfaffian computation~\cite{ferrari2006domino}.

\subsection{Main Results}
In this paper we show that there exist boundary conditions for which Glauber
dynamics mixes slowly for the six-vertex model in the ferroelectric and
antiferroelectric phases.  We start by proving that there are boundary conditions
that cause Glauber dynamics to be slow for all Boltzmann weights that lie in
the ferroelectric region of the phase diagram, where the mixing time is
exponential in the number of vertices in the lattice.
This is the first rigorous result for the mixing time of Glauber dynamics in
the ferroelectric phase and it gives a complete characterization.

\begin{restatable}[Ferroelectric Phase]{theorem}{ferroelectricThm}
\label{thm:ferroelectric}
For any $(a,b,c) \in \R^3_{> 0}$ such that
$a > b + c$ or $b > a + c$,
there exist boundary conditions for which Glauber
dynamics mixes exponentially slowly on~$\Lambda_{n}$.
\end{restatable}

\noindent
We note that our approach naturally breaks down at the critical line
of the conjectured phase diagram for the mixing time
in a way that
reveals a trade-off between the energy and entropy of the system.
Additionally, our analysis suggests an underlying
combinatorial interpretation for the phase transition between the ferroelectric 
and disordered phases in terms of the adherence strength of intersecting lattice paths
and the momentum parameter of correlated random walks.

Our second mixing result builds on the topological obstruction framework
developed in~\cite{ran-top} to show that Glauber dynamics with free boundary
conditions mixes slowly in most of the antiferroelectric region.
Specifically, we generalize the recent antiferroelectric mixing result
in~\cite{liu} with a Peierls argument that uses multivariate generating
functions for weighted non-backtracking walks instead of the connectivity
constant for (unweighted) self-avoiding walks to better account for
the discrepancies in Boltzmann weights.

\begin{restatable}[Antiferroelectric Phase]{theorem}{antiferroelectricThm}
\label{thm:antiferroelectric}
For any $(a,b,c) \in \R^3_{> 0}$ such that
$ac + bc + 3ab < c^2$,
Glauber dynamics mixes exponentially slowly on $\Lambda_{n}$
with free boundary conditions.
\end{restatable}

\noindent
We illustrate the new regions for which Glauber dynamics can be slowly mixing
in \Cref{fig:phases}.
Observe that our antiferroelectric subregion significantly extends Liu's and
pushes towards the conjectured threshold.

\subsection{Techniques}
We take significantly different approaches for our analysis of the ferroelectric and
antiferroelectric phases.
In the ferroelectric phase,
where $a > b + c$ and type-$a$ vertices are preferred to
type-$b$ and type-$c$ vertices,
we construct boundary conditions that induce
polynomially-many paths separated by a critical distance that
allows all of the paths to (1) behave independently and (2)
simultaneously intersect with their neighbors maximally.
(This analysis also covers the case $b > a + c$ by a standard invariant
that shows symmetry in the phase diagram over the line $y=x$.)
From here, we analyze the dynamics of a single path in isolation as an escape
probability, which eventually allows us to bound the conductance of the Markov
chain.
The dynamics of a single lattice path is equivalent to that
of a \emph{correlated random walk}. In \Cref{sec:tail-behavior}
we present a new tail inequality for
correlated random walks  that accurately bounds the
probability of large deviations from the starting position.
We note that decomposing the dynamics
of lattice models into one-dimensional random walks has recently been shown to
achieve nearly tight bounds for escape probabilities in a different
setting~\cite{durfee2018nearly}.

One of the key technical contributions in this paper is our analysis of the
tail behavior of correlated random walks in \Cref{sec:tail-behavior}.
While there is a simple combinatorial expression for the position of a
correlated random walk written as a sum of marginals, it is not immediately
useful for bounding the displacement from the origin.
To achieve an exponentially small tail bound for these walks, we first construct
a smooth function that tightly upper bounds the marginals and then optimize
this function to analyze the asymptotics of the log of the maximum marginal.
Once we obtain an asymptotic equality for the maximum marginal, we can upper
bound the deviation of a correlated random walk, and hence the deviation of a
lattice path in a configuration.
Ultimately, this allows us to show that there exists a balanced cut in the
state space that has an exponentially small escape probability, which implies
that the Glauber dynamics are slowly mixing.

In the antiferroelectric phase, on the other hand,
the weights satisfy $a + b < c$,
so {type-$c$} vertices are preferred.
It follows that there are two (arrow-reversal) symmetric ground states of
maximum probability containing only type-$c$ vertices.
To move between configurations that agree predominantly with different ground
states, the Markov chain must pass through configurations with a large number
of type-$a$ or type-$b$ vertices.
Using the idea of \emph{fault lines} introduced in \cite{ran-top},
we use \emph{weighted non-backtracking walks} to characterize such configurations
and construct a cut set with exponentially small
probability mass that separates the ground states.

\section{Preliminaries}
\label{sec:preliminaries}

We start with some background on
Markov chains and mixing times.
Let $\cM$ be an ergodic, reversible Markov chain with finite state space
$\Omega$, transition probability matrix $P$, and stationary distribution $\pi$.
The $t$-step transition probability from states $x$ to $y$ is denoted as $P^t(x,y)$.
The total variation distance between probability distributions
$\mu$ and $\nu$ on $\Omega$ is
\[
  \tv{\mu - \nu} = \frac{1}{2} \sum_{x \in \Omega}
    \abs*{\mu(x) - \nu(x)}.
\]
The \emph{mixing time} of $\cM$ is
$\tau\parens{1/4} = \min \set{t \in \Z_{\ge 0} : \max_{x \in \Omega} \tv{P^t(x, \cdot) - \pi} \le 1/4}$.
We say that $\cM$ is rapidly mixing if its mixing time is
$O(\poly(n))$,
where $n$ is the size of each configuration in the state space.
Similarly, we say that $\cM$ is slow mixing if its mixing time is
$\Omega(\exp(n^c))$ for some constant $c > 0$.

The mixing time of a Markov chain is characterized
by its \emph{conductance}
(up to polynomial factors).
The conductance of a nonempty set $S \subseteq \Omega$ is
\[
  \Phi\parens*{S} = \frac{\sum_{x \in S, y \not\in S} \pi(x) P(x,y) }{\pi(S)},
\]
and the conductance of the entire Markov chain is
$\Phi^* = \min_{S\subseteq \Omega : 0 < \pi(S) \le 1/2} \Phi(S)$.
It is often useful to view the conductance of a set as an escape
probability---starting from stationarity and conditioned on being in $S$,
the conductance $\Phi(S)$ is the probability that $\cM$ leaves $S$ in one step.

\begin{theorem}[\cite{levin2017markov}]
\label{thm:mixing-conductance-bound}
For an ergodic, reversible Markov chain with conductance $\Phi^*$, 
$\tau\parens{1/4} \ge 1/(4 \Phi^*)$.
\end{theorem}
\noindent
To show that a Markov chain is slow mixing, it suffices to show
that the conductance is exponentially small.

\section{Slow Mixing in the Ferroelectric Phase}
\label{sec:ferroelectric}

We start with the ferroelectric phase where $a > b + c$ or $b > a + c$,
and we give a conductance-based argument to show that Glauber
dynamics can be slowly mixing in the entire ferroelectric region.
Specifically, we show that there exist boundary conditions that induce
an exponentially small, asymmetric bottleneck in the state space, revealing 
a natural trade-off between the energy and entropy in the system.
Viewing the six-vertex model in the intersecting lattice path interpretation
suggests how to plant polynomially-many paths in the grid
that can (1) be analyzed independently, while (2) being capable
of intersecting maximally.
This path independence makes our analysis tractable and allows us to
interpret the dynamics of a path as a \emph{correlated random walk},
for which we develop an exponentially small tail bound in~\Cref{sec:tail-behavior}.
Since conductance governs mixing times,
we show how to relate the expected maximum deviation of a correlated walk to the
conductance of the Markov chain and prove slow mixing.
In addition to showing slow mixing up to the conjectured threshold, a
surprising feature of our argument is that it potentially gives a combinatorial
explanation for the phase transition from the ferroelectric to disordered
phase. In particular, \Cref{lem:escape-cut-mass} demonstrates how the
parameters of the model delicately balance the probability mass of the Markov
chain.

We start by leveraging the invariance of the Gibbs measure
and the lattice path interpretation of the six-vertex model
to conveniently reparameterize the Boltzmann weights.
Recall that for a fixed boundary condition, the invariants of the model~\cite{bleher2009exact}
imply that $a = \sqrt{w_1 w_2}$.
Therefore, we set
$w_1 = \lambda^2$ and $w_2 = 1$ to ignore empty sites
while letting $a = \lambda$.
We also set $b = w_3 = w_4 = \mu$ and $c = w_5 = w_6 = 1$
so that the weight of a configuration only comes from straight segments and
intersections of neighboring lattice paths.

\subsection{Constructing the Boundary Conditions and Cut}
\label{subsec:escape-construction}

We begin with a few colloquial definitions for lattice paths that allow us to
easily construct the boundary conditions and
make arguments about the conductance of the Markov chain.
We call a $2n$-step, north-east lattice path $\gamma$ starting from $(0,0)$ a
\emph{path of length $2n$}, and if the path ends at $(n,n)$ we describe it as
\emph{tethered}.
If $\gamma = ((0,0), (x_1,y_1),(x_2,y_2),\dots,(x_{2n},y_{2n}))$,
we define the \emph{deviation} of $\gamma$ to be
$\max_{i=0..2n} \norm{(x_i,y_i) - (i/2,i/2)}_{1} $.
Geometrically, path deviation captures the (normalized) maximum perpendicular
distance of the path to the line $y=x$.
We refer to vertices $(x_i,y_i)$ along the path as \emph{corners} or
\emph{straights} depending on whether or not the path turned.
If two paths intersect at a vertex
we call this site a \emph{cross}.
Note that this classifies all vertex types in the six-vertex model.

We consider the following \emph{independent paths
boundary condition} for an $n\times n$ six-vertex model
for the rest of the section.
To construct this boundary condition, we consider its lattice path
interpretation.
First, place a tethered path $\gamma_0$ that enters $(0,0)$ horizontally
and exits $(n,n)$ horizontally.
Next, place $2\ell = 2\floor{n^{1/8}}$ translated tethered paths of varying
length above and below the main diagonal,
each separated from its neighbors by distance $d = \floor{32n^{3/4}}$.
Specifically, the paths $\gamma_1, \gamma_2, \dots, \gamma_{\ell}$ below the main diagonal
begin at the vertices $(d,0), (2d,0),\dots,(\ell d, 0)$
and end at the vertices $(n,n-d), (n,n-2d),\dots, (n, n-\ell d)$, respectively.
The paths $\gamma_{-1},\gamma_{-2},\dots,\gamma_{-\ell}$ above the main diagonal
begin at $(0,d), (0,2d), \dots, (0,\ell d)$
and end at $(n-d,n), (n-2d,n), \dots, (n-\ell d, n)$.
The deviation of a translated tethered path is the deviation
of the same path starting at $(0,0)$.
To complete the boundary condition, we force the
paths below the main diagonal to enter vertically and exit horizontally.
Symmetrically, we force the paths above the main diagonal to enter horizontally
and exit vertically.
See \Cref{fig:escape-entropy} for an illustration of the
construction when all paths have small deviation.

\begin{figure}
\centering
\begin{subfigure}{0.5\textwidth}
  \centering
  \includegraphics[width=0.70\linewidth]{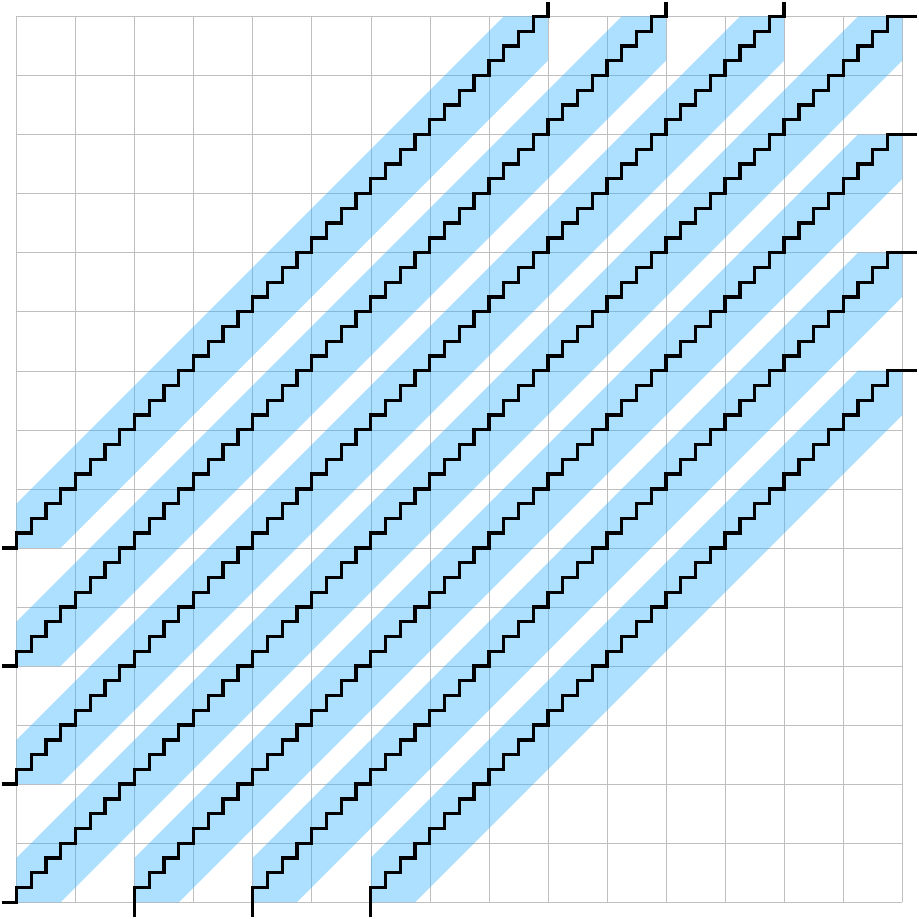}
  \caption{}
  \label{fig:escape-entropy}
\end{subfigure}%
\begin{subfigure}{0.5\textwidth}
  \centering
  \includegraphics[width=0.70\linewidth]{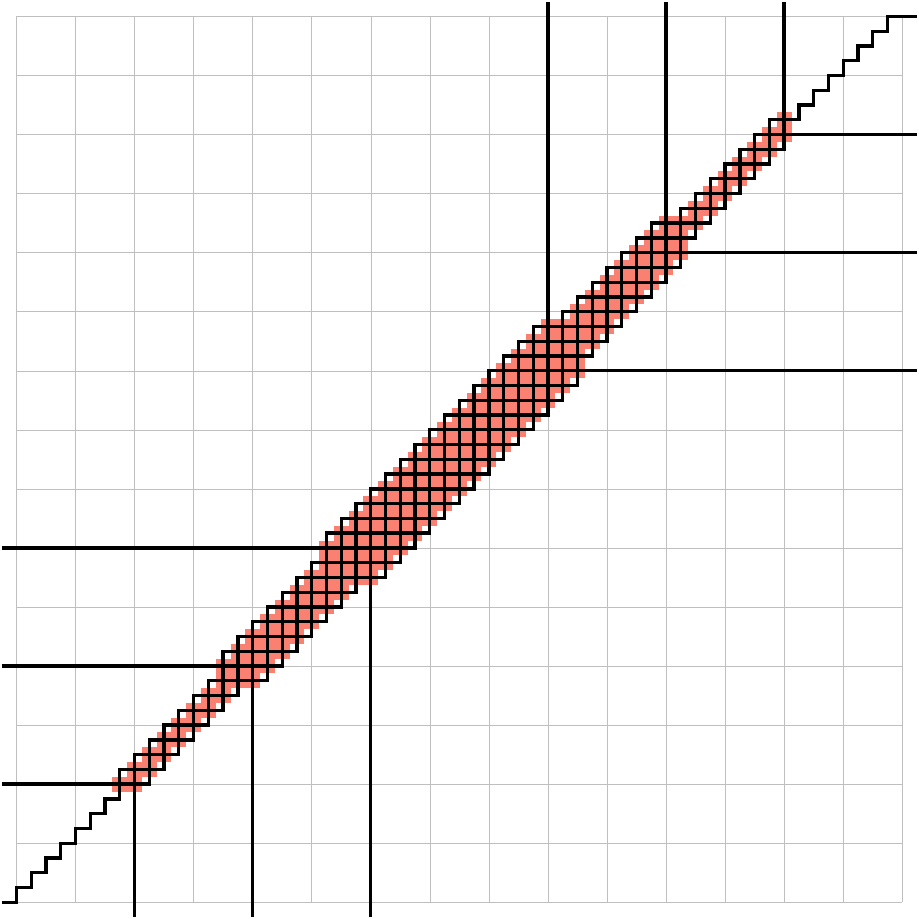}
  \caption{}
  \label{fig:escape-energy}
\end{subfigure}
  \caption{Examples of states with the independent paths boundary condition:
  (a) is a state in $S$ with the deviation bounds highlighted
  and (b) is the ground state in the ferroelectric phase.}
\label{fig:escape}
\end{figure}

Next, we construct an asymmetric cut in the state space induced by this boundary
condition in terms of its internal lattice paths.
In particular, we analyze a set $S$ of configurations such that
every path in a configuration has small deviation.  Formally, we let
\[
  S \DEF \set*{x \in \Omega : \text{the deviation of each path in $x$ is less than $8n^{3/4}$}}.
\]
Observe that by our choice of separation distance $d = \floor{32n^{3/4}}$
and the deviation limit for~$S$, no paths in any configuration of $S$
intersect.
It follows that the partition function for $S$ factors
into a product of $2\ell + 1$ partition functions,
one for each path with bounded deviation.
This intuition is useful when analyzing the conductance $\Phi(S)$ as an
escape probability from stationarity.

\subsection{Lattice Paths as Correlated Random Walks}
\label{sec:paths-and-walks}

Now we weight the internal paths according to the parameters
of the six-vertex model defined in the beginning of \Cref{sec:ferroelectric}.
The main result in this subsection is \Cref{lem:tethered-path-tail},
which states that random tethered
paths are exponentially unlikely to deviate past $\omega(n^{1/2})$,
even if drawn from a Boltzmann distribution that favors straights.
Start by defining $\Gamma(\mu, n)$ to be the distribution over 
tethered paths of length $2n$ with the property that
\[
  \Prob{\gamma} \propto \mu^{\text{($\#$ of straights in $\gamma$)}}.
\]

\begin{restatable}[]{lemma}{tetheredPathTail}
\label{lem:tethered-path-tail}
Let $\mu, \varepsilon > 0$ and $m = o(n)$.
For $n$ sufficiently large and $\gamma \sim \Gamma(\mu,n)$,
we have
\begin{align*}
  \Pr\parens*{\textnormal{$\gamma$ deviates by at least $2m$}}
  \le e^{-(1-\varepsilon)\frac{m^2}{\mu n}}.
\end{align*}
\end{restatable}

\noindent
Before giving the proof of \Cref{lem:tethered-path-tail},
we first introduce the concept of correlated random walks.
Then we present three prerequisite results about correlated random walks
and briefly explain their connection to the deviation of biased tethered paths.
Our goal here is to show how the supporting lemmas interact prior to the 
proof of \Cref{lem:tethered-path-tail}.

A key idea in our analysis of the ferroelectric phase is the notion of a \emph{correlated random
walk}, which generalize a simple symmetric random walk by accounting for
momentum.  A correlated random walk with
momentum parameter $p \in [0,1]$ starts at the origin and is defined as follows.
Let $X_1$ be a uniform random variable with support $\{-1,1\}$.
For all subsequent steps $i \ge 2$,
the direction of the process is correlated with the direction of the previous
step and satisfies
\begin{equation*}
  X_{i + 1} = \begin{cases}
    X_{i} & \text{with probability $p$,}\\
    -X_{i} & \text{with probability $1-p$}.
  \end{cases}
\end{equation*}
We denote the position of the walk at time $t$ by $S_{t} = \sum_{i=1}^t X_{i}$.
It will often be useful to make the change of variables $p=\mu/(1+\mu)$
when analyzing the six-vertex model,
where $\mu > 0$ is the weight of a straight vertex.
In many cases this also leads to cleaner expressions.
We use the following probability mass function (PMF) for the position of a
correlated random walk to develop our new tail inequality
(\Cref{lem:correlated_tail_bound}), which holds for all values of~$p$.
\begin{restatable}[\cite{hanneken1998exact}]{lemma}{exactPos}
\label{lem:exact_pos}
For any $n \ge 1$ and $m \ge 0$, the PMF of a
correlated random walk is
\begin{align*}
  \Prob{S_{2n}=2m}
  = \begin{cases}
    \frac{1}{2}p^{2n-1} & \text{if $2m=2n$,}\\
    \sum_{k=1}^{n-m} \binom{n+m-1}{k-1}\binom{n-m-1}{k-1} (1-p)^{2k-1}
       p^{2n-1-2k}
       \parens*{\frac{n(1-p)+k(2p-1)}{k}}
       & \text{if $2m < 2n$}.
  \end{cases}
\end{align*}
\end{restatable}

Now that we have defined correlated random walks,
we proceed by observing that there is a natural measure-preserving bijection between
biased tethered paths of length $2n$ and correlated random walks of length
$2n$ that return to the origin.
To see this, observe that every vertical edge in the tethered path corresponds
to a step to the right in the correlated random walk (i.e., $X_i = 1$),
and every horizontal edge in the tethered path corresponds to a step to the left in the correlated
random walk (i.e., $X_i = -1$).
Concretely, for a correlated random walk
$(S_0, S_1, \dots, S_{2n})$
parameterized by $p = \mu/(1+\mu)$, we have
\begin{equation}
\label{eqn:bijection}
  \Prob{\text{$\gamma$ deviates by at least $2m$}}
  =
  \ProbCond{\max_{i=0..2n} \abs*{S_i} \ge 2m}{S_{2n} = 0}.
\end{equation}

The first prerequisite lemma we present is an asymptotic equality
that generalizes the return probability of
simple symmetric random walks.
This allows us to relax the condition in \Cref{eqn:bijection} where the
correlated random walk must return to the origin,
and instead we bound
$\Prob{\max_{i=0..2n} |S_i| \ge 2m}$
at the expense of an polynomial factor.

\begin{lemma}[\cite{gillis1955correlated}]
\label{lem:return_prob}
For any constant $\mu > 0$, the return probability of a correlated random walk~is
\[
  \Prob{S_{2n}=0} \sim \frac{1}{\sqrt{\mu \pi n}}.
\]
\end{lemma}

The second result that we need in order to
prove \Cref{lem:tethered-path-tail} is that the PMF
for correlated random walks is monotone.

\begin{restatable}[]{lemma}{Unimodal}
\label{lem:unimodal}
  For any momentum parameter $p \in (0,1)$ and $n$ sufficiently large,
the probability of the position of a correlated random walk is monotone.
Concretely, for $m \in \{0,1,\dots,n-1\}$,
we have
\[
  \Prob{S_{2n} = 2m} \ge \Prob{S_{2n} = 2(m+1)}.
\]
\end{restatable}

\begin{proof}
We consider the cases $m=n-1$ and $m \in \{0, 1, 2, \dots, n-2\}$ separately.
Using \Cref{lem:exact_pos}, the probability density function for the
position of a correlated random walk is
\begin{align*}
  \Prob{S_{2n}=2m}
  = \begin{cases}
    \frac{1}{2}p^{2n-1} & \text{if $2m=2n$,}\\
    \sum_{k=1}^{n-m} \binom{n+m-1}{k-1}\binom{n-m-1}{k-1} (1-p)^{2k-1}
       p^{2n-1-2k}
       \parens*{\frac{n(1-p)+k(2p-1)}{k}}
       & \text{if $2m < 2n$}.
  \end{cases}
\end{align*}
If $m = n-1$, then we have the equations
\begin{align*}
  &\Prob{S_{2n} = 2m} = (1-p) p^{2n-3} \parens*{n(1-p) + 2p-1},\\
  &\Prob{S_{2n} = 2(m+1)} = \frac{1}{2} p^{2n - 1}.
\end{align*}
Therefore, we have $\Prob{S_{2n} = 2m} \ge \Prob{S_{2n} = 2(m+1)}$
for all
\[
  n \ge \frac{1}{1-p} \cdot \parens*{\frac{p^2}{2(1-p)} + 1 - 2p} > 0.
\]

Now we assume that $m \in \{0,1,2,\dots, n-2\}$.
Writing $\Prob{S_{2n}=2m} - \Prob{S_{2n} = 2(m+1)}$ as a difference of
sums and matching the corresponding terms,
it is instead sufficient to show for all
values of $k \in \{1,2,\dots, n-(m+1)\}$, we have
\[
  \binom{n+m-1}{k-1} \binom{n-m-1}{k-1} - \binom{n+(m+1)-1}{k-1}\binom{n-(m+1)-1}{k-1} \ge 0.
\]
Next, rewrite the binomial coefficients as
\begin{align*}
  \binom{n+(m+1)-1}{k-1} &= \frac{n+m}{n + m - (k - 1)} \cdot \binom{n+m-1}{k-1},\\
  \binom{n-(m+1)-1}{k-1} &= \frac{n - m - k}{n - m - 1} \cdot \binom{n-m-1}{k-1}.
\end{align*}
Therefore, it remains to show that
\[
  1 - \frac{n+m}{n+m-(k-1)} \cdot \frac{n-m-k}{n-m-1} \ge 0.
\]
Since all of the values in $\{n+m, n+m-(k-1), n-m-k, n-m-1\}$ are positive for
any choice of $m$ and $k$, it is equivalent to show that
\begin{align*}
  (n+m-(k-1))(n-m-1) \ge (n+m)(n-m-k).
\end{align*}
Observing that
\[
  (n+m-(k-1))(n-m-1) - (n+m)(n-m-k) = (2m+1)(k-1) \ge 0
\]
completes the proof.
\end{proof}

The third result we need is an upper bound for the position of a correlated random walk.
We fully develop this inequality in \Cref{sec:tail-behavior}
by analyzing the asymptotic behavior of the PMF in \Cref{lem:exact_pos}.
We note that \Cref{lem:correlated_tail_bound} shows exactly
how the tail behavior of simple symmetric random walks generalizes
to correlated random walks as a function of $\mu$.

\begin{restatable}[]{lemma}{CorrelatedTailBound}
\label{lem:correlated_tail_bound}
Let $\mu, \varepsilon > 0$ and $m = o(n)$.
For $n$ sufficiently large, a correlated random walk satisfies
\begin{align*}
\Prob{S_{2n} = 2m}
\le e^{-(1-\varepsilon)\frac{m^2}{\mu n}}.
\end{align*}
\end{restatable}

Now that we have established these supporting lemmas, we are
prepared to complete the proof of \Cref{lem:tethered-path-tail}, which also
heavily relies on union bounds and relaxing conditional probabilities.

\begin{proof}[Proof of \Cref{lem:tethered-path-tail}]
Using the measure-preserving bijection between
tethered paths of length $2n$ and correlated random walks
of length $2n$
(\Cref{sec:paths-and-walks})
along with the definition of conditional probability and
\Cref{lem:return_prob}, we have
\begin{align*}
  \Pr\parens*{\text{$\gamma$ deviates by at least $2m$}}
    &= \ProbCond{\max_{i=0..2n} \abs*{S_{i}} \ge 2m}{S_{2n} = 0} \\
    &\le \frac{\Prob{\max_{i=0..2n} \abs*{S_{i}} \ge 2m}}{\Prob{S_{2n} = 0}} \\
    &\le 2\sqrt{\mu \pi n} \cdot \Pr\parens*{\max_{i=0..2n} \abs*{S_{i}} \ge 2m},
\end{align*}
where the last inequality uses the definition of asymptotic equality
with $\varepsilon=1/2$.
Next, a union bound and the symmetry of correlated random walks imply that
\begin{align*}
  \Prob{\max_{i=0..2n} \abs*{S_{i}} \ge 2m}
  &\le \Prob{\max_{i=0..2n} S_i \ge 2m}
  + \Prob{\min_{i=0..2n} S_i \le -2m}\\
  &= 2 \cdot \Prob{\max_{i=0..2n} S_i \ge 2m}.
\end{align*}

Now we focus on the probability that the maximum 
position of the walk
is at least~$2m$.
For this event to be true,
the walk must reach $2m$ at some time
$i \in \{0, 1, 2, \dots, 2n\}$,
so by a union bound,
\begin{align*}
  \Prob{\max_{i=0..2n} S_i \ge 2m} 
    &\le \sum_{i=0}^{2n} \Prob{S_{i} = 2m}
    \le \sum_{i=1}^{n} \Prob{S_{2i} \ge 2m}.
\end{align*}
The second inequality takes into account the parity of the random
walk, the fact that if $i=0$ the walk can only be at position $0$,
and the relaxed condition that the final position is at least~$2m$.
\Cref{lem:unimodal} implies that the distribution is
unimodal on its support centered at the origin for sufficiently large $n$.
Moreover, for
walks of the same parity with increasing length and a fixed tail threshold,
the probability of the tail is nondecreasing.
Combining these two observations, we have
\begin{align*}
  \sum_{i=1}^{n} \Prob{S_{2i} \ge 2m}
  &\le n \cdot \Prob{S_{2n} \ge 2m}
  \le n^2 \cdot \Prob{S_{2n} = 2m}.
\end{align*}
Using the chain of previous inequalities
and the upper bound for $\Prob{S_{2n} = 2m}$ in \Cref{lem:correlated_tail_bound}
with the smaller error $\varepsilon/2$,
it follows that
\begin{align*}
  \Pr\parens*{\text{$\gamma$ deviates by at least $2m$}}
    &\le 4n^2\sqrt{\mu \pi n} \cdot
    \Prob{S_{2n}=2m}
  \le 
    e^{-(1-\varepsilon)\frac{m^2}{\mu n}},
\end{align*}
which completes the proof.
\end{proof}

\subsection{Bounding the Conductance and Mixing Time}
\label{subsec:escape-conductance}

Next, we bound the conductance of the Markov chain by viewing
$\Phi(S)$ as an escape probability.
We start by claiming that $\pi(S) \le 1/2$ (as required by the definition of conductance)
if and only if the parameters are in the ferroelectric phase.
Then we use the correspondence between tethered paths and correlated random walks
(i.e., \Cref{sec:paths-and-walks})
to prove that $\Phi(S)$ is exponentially small.

\begin{restatable}[]{lemma}{escapeCutMass}
\label{lem:escape-cut-mass}
Let $\mu > 0$ and $\lambda > 1 + \mu$ be constants.
For $n$ sufficiently large, $\pi\parens*{S} \le 1/2$.
\end{restatable}

\begin{proof}
We start by upper bounding $\pi(S)$ in terms of the partition function $Z$.
No paths in any state of $S$ deviate by more than $2n^{3/4}$ by the
definition of $S$.
Moreover, since adjacent paths are separated by distance $d = \floor{32n^{3/4}}$,
no two can intersect (\Cref{fig:escape-entropy}).
Therefore, it follows that the paths are independent of each other, which is convenient
because it allows us to implicitly factor the generating function for
configurations in $S$.

Next, observe that an upper bound for the generating function of any single path is
$\parens*{1+\mu}^{2n+1}$. 
This is true because all paths have length at most~$2n$, and we introduce an
additional $(1+\mu)^2$ factor to account for boundary conditions.
Since all the paths are independent and $\ell = \floor{n^{1/8}}$,
we have
\begin{align*}
  \pi\parens*{S} &\le \frac{\parens*{\parens*{1+\mu}^{2n+1}}^{2\ell + 1}}{Z}
    = \frac{\parens*{1+\mu}^{4n^{9/8}\parens*{1+o(1)}}\parens*{1+o(1)}}{Z}.
\end{align*}

Now we lower bound the partition function $Z$ of the entire model
by considering the weight of the ferroelectric
ground state (\Cref{fig:escape-energy}).
Recall that we labeled the $\ell$ paths below the main diagonal path
$\gamma_{1},\gamma_{2},\dots,\gamma_{\ell}$ such that
$\gamma_{\ell}$ is farthest from the main diagonal.
Let $c \le 10$ be a constant that accounts for subtle misalignments
between adjacent paths. It follows
that each path~$\gamma_k$ uniquely corresponds to
at least $n - (2kd + k + c)$ intersections.
Using the last path $\gamma_{\ell}$ as a lower bound for the number of
intersections that each path contributes and
accounting also for the paths above the main diagonal,
it follows that there are at least
\begin{align*}
\label{eqn:escape-intersections}
  2\ell\parens*{n - \parens*{2\ell d + \ell + c}}
  &= 2n^{9/8}\parens*{1 - o(1)}
\end{align*}
intersections in the ground state.

Similarly, we bound the number of straights that each path $\gamma_{k}$
contributes.
Note that we may also need an upper bound for this quantity in order to lower
bound the partition function since it is possible that $0 < \mu < 1$.
The number of straights in $\gamma_{k}$ is 
$2(kd \pm c)$, and $\gamma_0$ has two straights on the boundary.
Therefore, the total number of straights in the ground configuration is 
\begin{align*}
  2\sum_{k=1}^{\ell} 2\parens*{kd \pm c}
  &= 64n\parens*{1 + o(1)}.
\end{align*}

Since intersections are weighted by $\lambda^2$ and straights by $\mu$ in our
reparameterized model, by considering the ground state
and using the previous enumerations, it follows that
\begin{align*}
  Z \ge \parens*{\lambda^2}^{2n^{9/8}\parens*{1-o(1)}} 
    \mu^{64n\parens*{1+o(1)}}.
\end{align*}
Combining these inequalities allows us to upper bound the probability mass
of the cut $\pi(S)$ by
\begin{align*}
  \pi(S) &\le \parens*{\frac{1+\mu}{\lambda}}^{4n^{9/8}(1+o(1))}
      \mu^{-64n(1+o(1))} \parens*{1+o(1)}.
\end{align*}
Using the assumption that $\lambda > 1 + \mu$, we have
$\pi(S) \le 1/2$ for $n$ sufficiently large, as desired.
\end{proof}

Our analysis of the escape probability from $S$ critically relies
on the fact that paths in any state $x \in S$ are non-intersecting.
Combinatorially, we exploit the factorization of the generating
function for states in $S$ as a product of $2\ell + 1$ independent path
generating functions.

\begin{lemma}
\label{lem:escape-conductance}
Let $\mu, \varepsilon > 0$ be constants. For $n$ sufficiently large, 
$\Phi(S) \le e^{-(1-\varepsilon)\mu^{-1} n^{1/2}}$.
\end{lemma}

\begin{proof}
The conductance $\Phi(S)$ can be understood as the following escape
probability. Sample a state $x \in S$ from the
stationary distribution $\pi$ conditioned on $x \in S$, and run the
Markov chain from $x$ for one step to get a neighboring state $y$.
The definition of conductance
implies that $\Phi(S)$ is the probability that $y \not\in S$.
Using this interpretation, we can upper bound
$\Phi(S)$ by the probability mass of states that are near the boundary
of $S$ in the state space, since the process must escape in one step.
Therefore, it follows from the independent paths boundary condition and
the definition of $S$ that
\begin{align*}
  \Phi(S) &\le \ProbCond{\textnormal{there exists a path in $x$ deviating by at
  least $4n^{3/4}$}}{x \in S}.
\end{align*}

Next, we use a union bound over the $2\ell + 1$ different paths in a configuration
and consider the event that a particular path~$\gamma_k$ deviates by at least $4n^{3/4}$.
Because all of the paths in $S$ are independent,
we only need to consider the behavior of $\gamma_k$ in isolation.
This allows us to rephrase the conditional event.
Relaxing the conditional probability of each term in the sum gives
\begin{align*}
  \Phi(S) &\le \sum_{k=-\ell}^\ell
    \ProbCond{\text{$\gamma_k$ deviates by at least $4n^{3/4}$}}{x \in S}\\
  &= \sum_{k=-\ell}^\ell
    \ProbCond{\text{$\gamma_k$ deviates by at least $4n^{3/4}$}}{\text{$\gamma_k$ deviates by less than $8n^{3/4}$}}\\
&\le \sum_{k=-\ell}^\ell
    \frac{\Prob{\text{$\gamma_k$ deviates by at least $4n^{3/4}$}}}{
      1 - \Prob{\text{$\gamma_k$ deviates by at least $8n^{3/4}$}}}.
\end{align*}
\noindent
For large enough $n$, the length of every path $\gamma_k$ is in the range
$[n,2n]$ since we eventually have the inequality $n - \ell d \ge n/2$.
Therefore, we can apply \Cref{lem:tethered-path-tail}
with the error $\varepsilon/2$
to each term
and use the universal upper bound
\begin{align*}
 \frac{\Prob{\text{$\gamma_k$ deviates by at least $4n^{3/4}$}}}{
      1 - \Prob{\text{$\gamma_k$ deviates by at least $8n^{3/4}$}}} 
 &\le
  \frac{e^{-\left(1-\frac{\varepsilon}{2}\right)\frac{16 n^{3/2}}{\mu n}}}
  {1 - e^{-\left(1-\frac{\varepsilon}{2}\right)\frac{64 n^{3/2}}{\mu n}}}
 \le
  2 e^{-\parens*{1-\frac{\varepsilon}{2}} \frac{16n^{3/2}}{\mu n}}.
\end{align*}
It follows from the union bound and previous inequality
that the conductance $\Phi(S)$ is bounded by
\begin{align*}
  \Phi(S) &\le (2\ell + 1) \cdot 2 e^{-\parens*{1-\frac{\varepsilon}{2}} \frac{16n^{3/2}}{\mu n}}
  \le e^{-(1-\varepsilon)\mu^{-1} n^{1/2}},
\end{align*}
which completes the proof.
\end{proof}

Now that we have constructed a cut in the state space with exponentially small
conductance, we can obtain a bound on the mixing time when the probability mass
is properly distributed.

\begin{theorem}\label{thm:mainferro}
Let $\mu, \varepsilon > 0$ and $\lambda > 1+\mu$. For $n$ sufficiently large,
$\tau\parens*{1/4} \ge e^{(1-\varepsilon)\mu^{-1} n^{1/2}}$.
\end{theorem}

\begin{proof}
Since $\pi(S) \le 1/2$ by \Cref{lem:escape-cut-mass},
we have $\Phi^* \le \Phi(S)$.
The proof follows from \Cref{thm:mixing-conductance-bound} and
the conductance bound in \Cref{lem:escape-conductance}
with a smaller error $\varepsilon/2$.
\end{proof}

Last, we restate our main theorem and use \Cref{thm:mainferro} to show that
Glauber dynamics for the six-vertex model can be
slow mixing for all parameters in the ferroelectric phase.

\ferroelectricThm*

\begin{proof}
Without loss of generality, we
reparameterized the model so that $a = \lambda$, $b = \mu$, and $c = 1$.
Therefore, Glauber dynamics with the independent
paths boundary condition is slow mixing if $a > b + c$
by \Cref{thm:mainferro}.
Since the rotational invariance of the six-vertex model implies that
$a$ and $b$ are interchangeable parameters,
this mixing time result also holds in the case $b > a + c$.
\end{proof}

\section{Slow Mixing in the Antiferroelectric Phase}
\label{sec:antiferroelectric}

While Glauber dynamics can be slowly mixing in the ferroelectric phase,
we find it is true for substantially different reasons.
In the antiferroelectric phase, Boltzmann weights satisfy $a + b < c$,
so configurations tend to favor corner (i.e.,~type-$c$) vertices. 
The main insight behind our slow mixing proof is that when
$c$ is sufficiently large,
the six-vertex model can behave like the low-temperature hardcore model on $\Z^2$
where configurations predominantly agree with one of two ground states.
Liu recently formalized this argument in~\cite{liu} and showed that
Glauber dynamics for the six-vertex model with free boundary conditions
requires exponential time
when $\max(a,b) < \mu c$, where $\mu \le 2.639$ is the connective constant of
self-avoiding walks on the square lattice \cite{guttmann2001square}.
His proof uses a Peierls argument based on topological obstructions introduced
by Randall~\cite{ran-top} in the context of independent sets.
In this section, we extend Liu's result to the
region depicted in \Cref{fig:phase-diagram-us} by computing a closed-form
multivariate generating function that upper bounds the number of self-avoiding
walks and better accounts for disparities in their Boltzmann weights induced by
the parameters of the six-vertex model.

\subsection{Topological Obstruction Framework}
\label{subsec:setup}

We start with a recap of the definitions and framework laid out in~\cite{liu}.
There are two ground states in the antiferroelectric phase
such that every interior vertex is a corner:
$\xR$ (\Cref{fig:af-red}) and $\xG$ (\Cref{fig:af-green}).
These configurations are edge reversals of each other, so
for any state $x \in \Omega$ we can color its edges
\emph{red} if they are oriented as in $\xR$ or
\emph{green} if they are oriented as in $\xG$.
See \Cref{fig:af-mixed} for an example of how a configuration is colored.
It follows from case analysis of the six vertex types in \Cref{fig:six} that
the number of red edges incident to any internal vertex is even, and
if there are only two red edges then they must be rotationally adjacent to each other.
The same property holds for green edges by symmetry.
Note that the four edges bounding a cell of the lattice are monochromatic if and only if
they are oriented cyclically, and thus reversible by Glauber dynamics.
We say that a simple path from a horizontal edge on the left boundary of
$\Lambda_n$ to a horizontal edge on the right boundary
is a \emph{red horizontal bridge} if it contains only red edges.
We define green horizontal bridges and monochromatic vertical
bridges similarly.
A configuration has a \emph{red cross}
if it contains both a red horizontal bridge and a red vertical bridge.
Likewise, we can define a \emph{green cross}.
Let $\CR \subseteq \Omega$ be the set of all states with
a red cross, and let $\CG \subseteq \Omega$ be the set of all states with a
green cross.
It follows from \Cref{lem:liu-faultline} that $\CR \cap \CG = \emptyset$.

\begin{figure}
\centering
\begin{subfigure}{0.33\textwidth}
  \centering
  \includegraphics[width=0.75\linewidth]{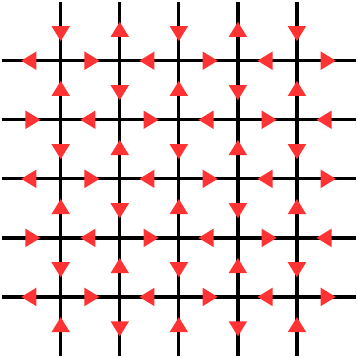}
  \caption{}
  \label{fig:af-red}
\end{subfigure}%
\begin{subfigure}{0.33\textwidth}
  \centering
  \includegraphics[width=0.75\linewidth]{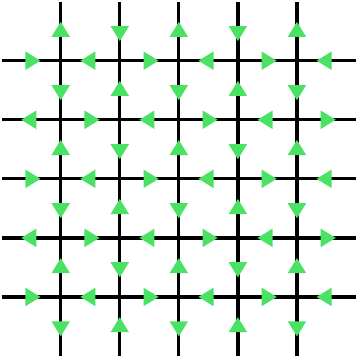}
  \caption{}
  \label{fig:af-green}
\end{subfigure}
\begin{subfigure}{0.33\textwidth}
  \centering
  \includegraphics[width=0.75\linewidth]{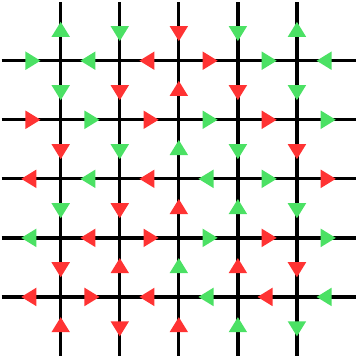}
  \caption{}
  \label{fig:af-mixed}
\end{subfigure}

\caption{Edge colorings of (a) the red ground state $\xR$,
  (b) the green ground state $\xG$, and
  (c) an example configuration with free boundary conditions that does not have a monochromatic cross.}
\label{fig:af-states}
\end{figure}

Next, we define the dual lattice $L_n$ to describe
configurations in $\Omega \setminus (\CR \cup \CG)$.
The vertices of~$L_n$ are the centers of the cells in $\Lambda_n$,
including the cells on the boundary that are partially enclosed,
and we connect dual vertices by an edge if their corresponding cells are
diagonally adjacent. Note that $L_n$ is a union
of two disjoint graphs (\Cref{fig:af-dual}).
For any state~$x \in \Omega$ there is a corresponding dual subgraph
$L_x$ defined as follows:
for each interior vertex $v$ in $\Lambda_n$, 
if $v$ is incident to two red edges and two green edges,
then~$L_x$ contains the dual edge
passing through $v$ that
separates the two red edges from the two green edges.
This construction is well-defined because the red edges are rotationally
adjacent.
See \Cref{fig:af-example-dual} for an example of a dual configuration.
For any $x \in \Omega$, we say that $x$ has a \emph{horizontal fault line}
if $L_x$ contains a simple path from a left dual boundary vertex to a right
dual boundary vertex.
We define horizontal fault lines similarly 
and let $\CFL \subseteq \Omega$ be the set of all states containing a
horizontal or vertical fault line.
Fault lines completely separate red and green edges, and hence
are topological obstructions that prohibit monochromatic bridges.

\begin{figure}
\centering
\begin{subfigure}{0.33\textwidth}
  \centering
  \includegraphics[width=0.75\linewidth]{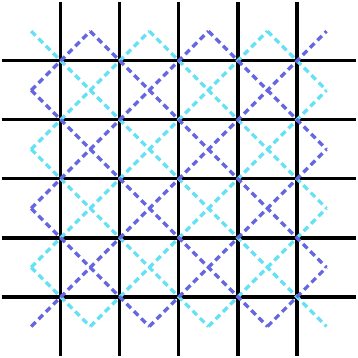}
  \caption{}
  \label{fig:af-dual}
\end{subfigure}%
\begin{subfigure}{0.33\textwidth}
  \centering
  \includegraphics[width=0.75\linewidth]{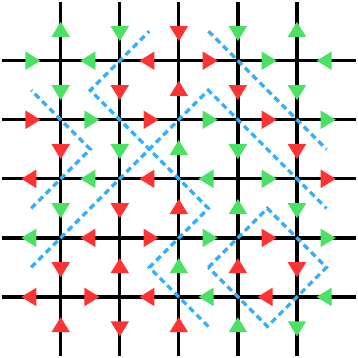}
  \caption{}
  \label{fig:af-example-dual}
\end{subfigure}
\begin{subfigure}{0.33\textwidth}
  \centering
  \includegraphics[width=0.75\linewidth]{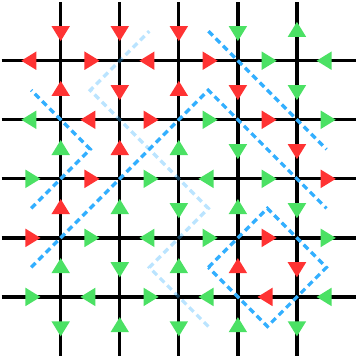}
  \caption{}
  \label{fig:af-map}
\end{subfigure}
\caption{
  Illustrations of (a) the dual lattice $L_n$ as a union of disjoint cyan and purple subgraphs,~(b) an example configuration overlaid with its dual graph,
  and (c) the example under the injective fault line map.
  }
\label{fig:af-injection}
\end{figure}

Last, we extend the notion of fault lines to \emph{almost fault lines}.
We say that $x \in \Omega$ has a horizontal almost fault line
if there is a simple path in $L_n$ connecting a left dual boundary vertex to a right
dual boundary vertex such that all edges except for one are in $L_x$.
We define vertical almost fault lines similarly
and let the set $\CAFL \subseteq \Omega$
denote all states containing an almost fault line.
Finally, let $\partial \CR \subseteq \Omega$ denote the set of states 
not in $\CR$ that one move away from $\CR$ in the state space
according to the Glauber dynamics.

\begin{lemma}[\cite{liu}]
\label{lem:liu-faultline}
We can partition the state space into
$\Omega = \CR \cup \CFL \cup \CG$.
Furthermore,
we have $\partial \CR \subseteq \CFL \cup \CAFL$.
\end{lemma}

\subsection{Bounding the Mixing Time with a Peierls Argument}
\label{subsec:extension}

In this subsection we show that $\pi(\CFL \cup \CAFL)$ is an exponentially small bottleneck
in the state space $\Omega$.
The analysis relies on \Cref{lem:liu-faultline} and
a new multivariate upper bound for weighted self-avoiding walks
(\Cref{lem:refined_upper_bound}).
Our key observation is that when a fault line changes direction,
the vertices in its path change from type-$a$ to
type-$b$ or vice versa.
Therefore, our goal in this subsection is to generalize the trivial
$3^{n-1}$ upper bound for the number of self-avoiding walks
by accounting for their changes in direction in aggregate.
We achieve this by using generating functions to
solve a system of linear recurrence relations.

We start by encoding \emph{non-backtracking walks} that start from the origin
and take their first step northward
using the characters in $\{\texttt{S},\texttt{L},\texttt{R}\}$,
representing straight, left, and right steps.
For example, the walk $\texttt{SLRSSL}$
corresponds to the sequence
$((0,0),(0,1),(-1,1),(-1,2),(-1,3),(-1,4),(-2,4))$.
If a fault line is the same shape as
$\texttt{SLRSSL}$ up to a rotation about the origin, then there
are only two possible sequences of vertex types through which it can pass:
$abaaab$ and $babbba$.
This follows from the fact that once the first vertex type is determined,
only turns in the self-avoiding walk
(i.e., the $\texttt{L}$ and $\texttt{R}$ characters)
cause the vertex type to switch.
We define the weight of a fault line to be the product
of the vertex types through which it passes.
More generally, we define the weight of a non-backtracking walk
that initially passes through a fixed vertex type
to be the product of the induced vertex types according to the rule that
turns toggle the current type.
Formally, we let
the function $g_a(\gamma) : \{\texttt{S}\}\times\{\texttt{S},\texttt{L},\texttt{R}\}^{n-1} \rightarrow \R_{\ge 0}$
denote the weight of a non-backtracking walk~$\gamma$
that starts by crossing a type-$a$ vertex.
We define the function $g_b(\gamma)$ similarly and provide the examples
$g_a(\texttt{SLRSSL}) = a^4b^2$ and $g_b(\texttt{SLRSSL}) = a^2b^4$ for clarity.
Last, observe that a sequence of vertex types can have many different walks
in its preimage.
The non-backtracking walk $\texttt{SRRSSR}$ also maps to
$abaaab$ and $babbba$---in fact, there are $2^3 = 8$ such walks in this example
since we can interchange $\texttt{L}$ and~$\texttt{R}$ characters.

The idea of enumerating the preimages of a binary string corresponding to
sequence of vertex types suggests a recursive approach for computing 
the sum of weighted non-backtracking walks.
This naturally leads to the use of generating functions, so we overload the
variables $x$ and $y$ to also denote function arguments.
For nonempty binary string $s \in \{0,1\}^n$, let~$h(s)$ count the number of
pairs of adjacent characters that are not equal and let $|s|$ denote the
number of ones in $s$ (e.g., if $s = 010001$ then $h(s) = 3$ and $|s| = 2$).
The sum of weighted self-avoiding walks is upper bounded by the sum of
weighted non-backtracking walks, so we proceed by analyzing the following
function:

\begin{equation}
\label{eqn:F_def}
  F_n(x,y) \DEF \sum_{\gamma \in \{\texttt{S}\}\times\{\texttt{S},\texttt{L},\texttt{R}\}^{n-1}} g_x(\gamma) + g_y(\gamma)
    = \sum_{s \in \{0,1\}^n} 2^{h(s)} x^{|s|} y^{n - |s|}.
\end{equation}

\noindent
Note that $F_n(1,1) = 2 \cdot 3^{n-1}$ recovers the number of non-backtracking
walks that initially cross type-$a$ or type-$b$ vertices.

In the next section, we compute the closed-form solution for $F_n(x,y)$
by diagonalizing a matrix corresponding to the system of recurrence
relations, which allows us to accurately quantify the discrepancy between fault
lines when the Boltzmann weights $a$ and~$b$ differ.
For now, we use the following upper bound for $F_n(x,y)$ in our Peierls
argument and defer its proof to \Cref{app:antiferroelectric}.

\begin{restatable}{lemma}{refinedUpperBound}
\label{lem:refined_upper_bound}
Let $F_n(x,y)$ be the generating function for weighted non-backtracking walks
defined in \Cref{eqn:F_def}.
For any integer $n \ge 1$ and $x,y \in \R_{> 0}$, we have
\[
  F_{n}(x,y) \le 3\parens*{x + y} \parens*{\frac{x + y + \sqrt{x^2 + 14xy + y^2}}{2}}^{n-1}.
\]
\end{restatable}

The first step of our Peierls argument is to upper bound $\pi(\CFL \cup \CAFL)$,
which then gives us a bound on the conductance
and allows us to prove \Cref{thm:antiferroelectric}.
We start by defining the subset of antiferroelectric parameters that cause
$F_n(a/c,b/c)$ to decrease exponentially fast.

\begin{restatable}{lemma}{paramsLessThanOne}
\label{lem:params_less_than_one}
If $(a,b,c) \in \R_{> 0}^3$ is antiferroelectric and
$3ab + ac + bc < c^2$, then
\[
  a + b + \sqrt{a^2 + 14ab + b^2} < 2c.
\]
\end{restatable}

\begin{proof}
Let $x = a/c$ and $y = b/c$,
and observe that $0 < x < 1$ by the antiferroelectric assumption.
It follows from our hypothesis that $y < (1-x)/(1 + 3x)$.
Therefore, we have
\begin{align*}
  x + y + \sqrt{x^2 + 14xy + y^2} &<
    x + \frac{1-x}{1 + 3x} + \sqrt{\frac{x^2(1+3x)^2 + 14x(1-x)(1+3x) + (1-x)^2}{(1+3x)^2} }\\
  &= \frac{x(1+3x) + 1-x + \sqrt{\parens*{3x^2 - 6x - 1}^2}}{1+3x}\\
  &= \frac{x(1+3x) + 1-x - \parens*{3x^2 - 6x - 1}}{1+3x}\\
  &= \frac{2(1+3x)}{1+3x}\\
  &= 2,
\end{align*}
which completes the proof.
\end{proof}

\begin{lemma}
\label{lem:peierls}
If $(a,b,c) \in \R_{> 0}^3$ is antiferroelectric and
$3ab + ac + bc < c^2$, then for Glauber dynamics with free boundary conditions we have
\[
  \pi\left(\CFL \cup \CAFL\right) \le \poly(n) \parens*{\frac{a+b+\sqrt{a^2+14ab+b^2}}{2c}}^n.
\]
\end{lemma}

\begin{proof}
For any self-avoiding walk $\gamma$ and dual vertices $s, t \in L_n$ on the boundary,
let $\Omega_{\gamma,s,t} \subseteq \Omega$ be the set of
states that contain $\gamma$ as a fault line or an almost fault line
such that $\gamma$ starts at $s$ and ends at~$t$.
Without loss of generality, assume that the (almost) fault line is vertical.
Reversing the direction of all edges on the left side of $\gamma$ defines the
injective map
$f_{\gamma,s,t} : \Omega_{\gamma,s,t} \rightarrow \Omega \setminus \Omega_{\gamma,s,t}$
such that if $\gamma$ is a fault line of
$x \in \Omega_{\gamma,s,t}$, then the weight of its image
$f_{\gamma,s,t}(x)$ is amplified by $c^{|\gamma|}/g_a(\gamma)$
or~$c^{|\gamma|}/g_b(\gamma)$.
For an example of this injection, see \Cref{fig:af-map}.
Similarly, if~$\gamma$ is an almost fault line, decompose $\gamma$ into subpaths
$\gamma_1$ and $\gamma_2$ separated by a type-$c$ vertex
such that~$\gamma_1$ starts at $s$ and $\gamma_2$ ends at $t$.
In this case, the weight of the images of almost fault lines
is amplified by a factor of
$\min(a,b)/c \cdot c^{|\gamma_1| + |\gamma_2|}/(g_\alpha(\gamma_1)g_\beta(\gamma_2))$
for some $(\alpha,\beta) \in \{a,b\}^2$.
Using the fact that $f_{\gamma,s,t}$ is injective and summing over
the states containing $\gamma$ as a fault line and an almost fault line
separately gives us
\begin{align}
\label{eqn:saw_with_terminals_ineqality}
  \pi\parens*{\Omega_{\gamma,s,t}} \le
    \frac{g_a(\gamma) + g_b(\gamma)}{c^{|\gamma|}}
    +
    \frac{c}{\min(a,b)}\sum_{\gamma_1 + \gamma_2 = \gamma}
    \frac{g_a(\gamma_1)+g_b(\gamma_1)}{c^{|\gamma_1|}} \cdot
      \frac{g_a(\gamma_2)+g_b(\gamma_2)}{c^{|\gamma_2|}},
\end{align}
where the sum
is over all $\Theta(|\gamma|)$ decompositions of $\gamma$ into
$\gamma_1$ and $\gamma_2$.

Equipped with \Cref{eqn:saw_with_terminals_ineqality} and \Cref{lem:refined_upper_bound},
we use a union bound over all pairs of terminal vertices $(s,t)$
and fault line lengths $\ell$ to bound $\pi(\CFL \cup \CAFL)$
in terms of the generating function for weighted non-backtracking walks $F_\ell(x,y)$.
Since antiferroelectric weights satisfy $3ab+ac+bc<c^2$,
it follows from \Cref{lem:params_less_than_one} that
\begin{align*}
  \pi\parens*{\CFL \cup \CAFL} &\le \sum_{(s,t)} \sum_{\ell = n}^{n^2}
      \parens*{F_\ell\parens*{a/c,b/c} + \frac{c}{\min(a,b)}\sum_{k=0}^\ell F_{k}\parens*{a/c,b/c} F_{\ell-k}\parens*{a/c,b/c}}\\
    &\le \sum_{(s,t)} \sum_{\ell = n}^{n^2}
      \text{poly}(\ell) \parens*{\frac{a+b+\sqrt{a^2+14ab+b^2}}{2c}}^{\ell}\\
    &\le \text{poly}(n) \parens*{\frac{a+b+\sqrt{a^2+14ab+b^2}}{2c}}^{n}.
\end{align*}
Note that the convolutions in the first inequality generate all \emph{almost}
weighted non-backtracking walks.
\end{proof}

\antiferroelectricThm*

\begin{proof}[Proof of \Cref{thm:antiferroelectric}]
Let $\OMEGAMIDDLE = \CFL \cup \CAFL$, $\OMEGALEFT = \CR \setminus \OMEGAMIDDLE$,
and $\OMEGARIGHT = \CG \setminus \OMEGAMIDDLE$.
It follows from \Cref{lem:liu-faultline} that 
$\Omega = \OMEGALEFT \cup \OMEGAMIDDLE \cup \OMEGARIGHT$
is a partition with the properties that
$\partial \OMEGALEFT \subseteq \OMEGAMIDDLE$ and
$\pi(\OMEGALEFT) = \pi(\OMEGARIGHT)$.
Since the partition is symmetric, \Cref{lem:peierls} implies that
$1/4 \le \pi(\OMEGALEFT) \le 1/2$, for $n$ sufficiently large.
Therefore, we can upper bound the conductance by
$\Phi^* \le \Phi\parens*{\OMEGALEFT} \le 4\pi\parens*{\OMEGAMIDDLE}$.
Using \Cref{thm:mixing-conductance-bound} along with \Cref{lem:peierls} and
\Cref{lem:params_less_than_one} gives the desired mixing time bound.
\end{proof}

\subsection{Weighted Non-Backtracking Walks}
\label{app:antiferroelectric}

In this section we present a closed-form formula for
the weighted non-backtracking walks generating function $F_{n}(x,y)$,
and we give the proof of \Cref{lem:refined_upper_bound}.
We start by decomposing the generating function $F_n(x,y)$ into two sums
over disjoint sets of bit strings defined by their final character.
Formally, for any $n \ge 1$, let
\begin{align*}
  F_{n,0}(x,y) &= \sum_{s \in \{0,1\}^{n-1} \times \{0\}} 2^{h(s)} x^{|s|} y^{n - |s|} \\
  F_{n,1}(x,y) &= \sum_{s \in \{0,1\}^{n-1} \times \{1\}} 2^{h(s)} x^{|s|} y^{n - |s|}.
\end{align*}
First, note that $F_n(x,y) = F_{n,0}(x,y) + F_{n,1}(x,y)$.
Second, observe that by recording the final character of the bit strings, we can
design a system of linear recurrences to account for the $2^{h(s)}$ term
appearing in \Cref{eqn:F_def}, which counts the number of non-backtracking
walks that map to a given sequence of vertex types.

\begin{lemma}\label{lem:refined_recurrence}
For any integer $n \ge 1$ and $x,y \in \R_{> 0}$,
we have the system of recurrence relations
\begin{align*}
  F_{n+1,0}(x,y) &= xF_{n,0}(x,y) + 2x F_{n,1}(x,y)\\
  F_{n+1,1}(x,y) &= 2yF_{n,0}(x,y) + y F_{n,1}(x,y),
\end{align*}
where the base cases are $F_{1,0}(x,y) = x$ and $F_{1,1}(x,y) = y$.
\end{lemma}

\begin{proof}
This immediately follows from the definitions of
the functions $F_{n,0}(x,y)$ and $F_{n,1}(x,y)$.
\end{proof}

\begin{lemma}
\label{lem:refined_closed_form}
For any integer $n \ge 1$ and $x,y \in \R_{> 0}$, define the values
\begin{align*}
  m &= \sqrt{x^2 + 14xy + y^2}\\
  \lambda_1 &= \frac{1}{2}\parens*{x + y - m}\\
  \lambda_2 &= \frac{1}{2}\parens*{x + y + m}.
\end{align*}
The generating $F_n(x,y)$ can be written in closed-form as
\begin{align*}
  F_{n}(x,y) =
  \frac{1}{2m} \parens*{
    \parens*{x^2 + 6xy + y^2 + m(x + y)} \lambda_{2}^{n-1}
    -
    \parens*{x^2 + 6xy + y^2 - m(x + y)} \lambda_{1}^{n-1}
  }.
\end{align*}
\end{lemma}

\begin{proof}
For brevity, we let $F_{n,0} = F_{n,0}(x,y)$ and $F_{n,1} = F_{n,1}(x,y)$.
It follows from \Cref{lem:refined_recurrence} that
\begin{align*}
  \begin{bmatrix}
    F_{n+1,0} \\
    F_{n+1,1}
  \end{bmatrix}
  =
  \begin{bmatrix}
    x & 2x \\
    2y & y
  \end{bmatrix}
  \begin{bmatrix}
    F_{n,0} \\
    F_{n,1}
  \end{bmatrix}.
\end{align*}
Next, observe that the recurrence matrix is diagonalizable. In particular,
we have
\begin{align*}
  A = \begin{bmatrix} x & 2x \\ 2y & y \end{bmatrix}
    = P \Lambda P^{-1},
\end{align*}
where
\begin{align*}
  P = \frac{1}{4y}\begin{bmatrix}
    x - y - m & x - y + m\\
    4y & 4y
  \end{bmatrix} \hspace{0.75cm}
  \Lambda &= \begin{bmatrix}
   \lambda_1 & 0 \\
   0 & \lambda_2
 \end{bmatrix} \hspace{0.75cm}
  P^{-1} = \frac{1}{2m} \begin{bmatrix}
    -4y & x-y+m \\
    4y & -(x - y - m)
  \end{bmatrix}.
\end{align*}
Since the base cases are $F_{n,0} = x$ and $F_{n,1} = y$, it follows that
\begin{align*}
  \begin{bmatrix}
    F_{n,0} \\
    F_{n,1}
  \end{bmatrix}
  = A^{n-1}
    \begin{bmatrix}
      F_{1,0} \\
      F_{1,1}
    \end{bmatrix}
  = P \Lambda^{n-1} P^{-1}
  \begin{bmatrix}
    x \\ y
  \end{bmatrix}.
\end{align*}
Using the fact $F_{n}(x,y) = F_{n,0}(x,y) + F_{n,1}(x,y)$
and simplifying the matrix equation above gives us
\begin{align*}
  F_{n}(x,y) 
    &= \frac{1}{8my}
      \parens*{y(3x+y+m)(x+3y+m) \lambda_2^{n-1} - 
      y(3x+y-m)(x+3y-m) \lambda_1^{n-1}
      }\\
    &= \frac{1}{2m} \parens*{
      \parens*{x^2 + 6xy + y^2 + m(x+y)} \lambda_{2}^{n-1}
      -
      \parens*{x^2 + 6xy + y^2 - m(x+y)} \lambda_{1}^{n-1}
    },
\end{align*}
as desired.
\end{proof}

\refinedUpperBound*

\begin{proof}
We start by using \cref{lem:refined_closed_form} to rewrite the closed-form
solution of $F_{n}(x,y)$ as
\begin{align*}
  F_{n}(x,y) = \frac{1}{2m} \parens*{
    \parens*{x^2 + 6xy + y^2} \parens*{\lambda_{2}^{n-1} - \lambda_{1}^{n-1}}
    + m(x+y)\parens*{\lambda_2^{n-1} + \lambda_{1}^{n-1}}
  }.
\end{align*}
Next, we observe that the eigenvalue $\lambda_1$
satisfies $\lambda_1 < 0$ and $\abs{\lambda_1} \le \lambda_2$.
Since $(x+y)^2 < m^2$, it follows that
$x+y - m = 2\lambda_1 < 0$.
Furthermore,  we have
$2\abs{\lambda_1} \le \abs{x+y} + \abs{-m} = 2\lambda_2$ by the triangle inequality.
Together these two properties imply that
\begin{align*}
  \lambda_2^{n-1} - \lambda_1^{n-1} \le 2\lambda_{2}^{n-1}
  \hspace{0.75cm}\text{and}\hspace{0.75cm}
  \lambda_2^{n-1} + \lambda_1^{n-1} \le 2\lambda_{2}^{n-1}.
\end{align*}
Therefore, we can upper bound $F_{n}(x,y)$ by
\begin{align*}
  F_{n}(x,y) &\le \frac{1}{m}\parens*{x^2 + 6xy + y^2} \lambda_{2}^{n-1} + (x+y) \lambda_{2}^{n-1}.
\end{align*}
Since $x^2 + 6xy + y^2 < m^2$, we have the inequalities
\begin{align*}
  \frac{1}{m}(x^2 + 6xy + y^2) &< \sqrt{x^2 + 6xy + y^2} 
  < \sqrt{(2x + 2y)^2} = 2(x+y).
\end{align*}
The result follows from the definition of $\lambda_2$.
\end{proof}

\section{Tail Behavior of Correlated Random Walks}
\label{sec:tail-behavior}

In this section we prove~\Cref{lem:correlated_tail_bound}, which gives an
exponentially small upper bound for the tail of a correlated random walk as a
function of its momentum parameter $\mu$.  Our proof builds off of the PMF for
the position of a correlated random walk restated below, which is combinatorial
in nature and not readily amenable for tail inequalities.
Specifically, the probability $\Prob{S_{2n}=2m}$ is a sum of marginals
conditioned on the number of turns that the walk
makes~\cite{renshaw1981correlated}.

\exactPos*

There are two main ideas in our approach to develop a more useful bound for
the position  of a correlated random walk $\Prob{S_{2n}=2m}$.  First, we
construct a smooth function that upper bounds the marginals as a function of
$x$ (a continuation of the number of turns in the walk $k$),
and then we determine its maximum value.
Next we show that the log of the maximum value is asymptotically equivalent to
$m^2/(\mu n)$ for $m = o(n)$, which gives us desirable bounds for
sufficiently large values of~$n$.
We note that our analysis illustrates precisely how correlated random walks
generalize simple symmetric random walks and how the momentum parameter
$\mu$ controls the exponential decay.

\subsection{Upper Bounding the Marginal Probabilities}

We start by using Stirling's approximation to construct a smooth function
that upper bounds the marginal terms in the sum of the PMF for correlated
random walks.
For $x \in (0, n-m)$, let
\begin{align}
\label{eqn:f_def}
  f(x) \DEF
  \begin{cases}
    1 & \text{if $x=0$,}\\
  \frac{(n+m)^{n+m}}{x^x (n+m-x)^{n+m-x}} \cdot
  \frac{(n-m)^{n-m}}{x^x (n-m-x)^{n-m-x}} \cdot  \mu^{-2x} & \text{if $x \in (0,n-m)$,}\\
    \mu^{-2(n-m)} & \text{if $x = n-m$}.
  \end{cases}
\end{align}
It can easily be checked that $f(x)$ is continuous on all of $[0,n-m]$
using the fact that $\lim_{x\rightarrow 0} x^x = 1$.

\begin{lemma}
For any integer $m \ge 0$, a correlated random walk satisfies
\label{lem:proxy_upperbound}
\[
  \Prob{S_{2n}=2m} \le \poly(n) \sum_{k=0}^{n-m} \parens*{\frac{\mu}{1+\mu}}^{2n} f(k).
\]
\end{lemma}

\begin{proof}
Consider the probability density function for $\Prob{S_{2n}=2m}$
in \Cref{lem:exact_pos}.
If $2m=2n$ the claim is clearly true, so we
focus on the other case.
We start by bounding the rightmost polynomial term in the sum.
For all $n \ge 1$, we have
\begin{align*}
  \frac{n(1-p)+k(2p-1)}{k} \le 2n.
\end{align*}
Next, we reparameterize the marginals in terms of $\mu$,
where $p = \mu/(1+\mu)$,
and use a more convenient upper bound for the binomial coefficients.
Observe that
\begin{align*}
  \Prob{S_{2n}=2m} &\le
    2n\sum_{k=1}^{n-m} \binom{n+m-1}{k-1}\binom{n-m-1}{k-1}
      \parens*{\frac{1}{1+\mu}}^{2k-1} \parens*{\frac{\mu}{1+\mu}}^{2n-1-2k}\\
  &\le \poly(n) 
    \sum_{k=0}^{n-m} \binom{n+m}{k}\binom{n-m}{k}
    \parens*{\frac{\mu}{1+\mu}}^{2n}
    \mu^{-2k}.
\end{align*}
Stirling's approximation states that for all $n \ge 1$ we have
\[
  e \parens*{\frac{n}{e}}^n \le n! \le en\parens*{\frac{n}{e}}^n,
\]
so we can bound the products of binomial coefficients
up to a polynomial factor by
\begin{align*}
  \binom{n+m}{k}\binom{n-m}{k} 
    &\le \text{poly}(n) \cdot
    \frac{\parens*{\frac{n+m}{e}}^{n+m}}{\parens*{\frac{k}{e}}^k \parens*{\frac{n+m-k}{e}}^{n+m-k} }
    \cdot
    \frac{\parens*{\frac{n-m}{e}}^{n-m}}{\parens*{\frac{k}{e}}^k \parens*{\frac{n-m-k}{e}}^{n-m-k} }\\
  &= \text{poly}(n) \cdot
  \frac{(n+m)^{n+m}}{k^k (n+m-k)^{n+m-k}} \cdot
  \frac{(n-m)^{n-m}}{k^k (n-m-k)^{n-m-k}}.
\end{align*}
The proof follows the definition of $f(x)$ given in \Cref{eqn:f_def}.
\end{proof}

There are polynomially-many marginal terms in the sum of the PMF, so
if the maximum term is exponentially small, then 
the total probability is exponentially small.
Since the marginal terms are bounded above by an expression involving $f(x)$,
we proceed by maximizing $f(x)$ on its support.

\begin{lemma}
\label{lem:critical_point}
The function $f(x)$ is maximized at the critical point
\begin{align*}
    x^* = \begin{cases}
      \frac{n^2 - m^2}{2n} & \text{if $\mu = 1$,}\\
      \frac{n}{1-\mu^2}\parens*{1-\sqrt{\mu^2+(1-\mu^2)\frac{m^2}{n^2}}} & \text{otherwise.}
    \end{cases}
\end{align*}
\end{lemma}

\begin{proof}
We start by showing that $f(x)$ is log-concave on $(0,n-m)$, which
implies that it is unimodal. It follows that
a local maximum of $f(x)$ is a global maximum.
Since $n$ and $k$ are fixed as constants and because the numerator is positive,
it is sufficient to show that
\begin{align*}
  g(x) &= -\log\parens*{x^x (n+m-x)^{n+m-x} \cdot x^x (n-m-x)^{n-m-x} \cdot \mu^{2x}}\\
  &= - \parens*{2x \log(\mu x) + (n+m-x)\log(n+m-x) + (n-m-x)\log(n-m-x)}
\end{align*}
is concave.
Observe that the first derivative of $g(x)$ is
\begin{align*}
  g'(x) &= -2\parens*{1+\log\parens*{\mu x}} + \parens*{1+\log(n+m-x)}
    + \parens*{1 + \log\parens*{n-m-x}}\\
    &= -2\log\parens*{\mu x} + \log\parens*{n+m-x} + \log\parens*{n-m-x},
\end{align*}
and the second derivative is
\begin{align*}
  g''(x) &= -\frac{2}{x} - \frac{1}{n+m-x} - \frac{1}{n-m-x}.
\end{align*}
Because $g''(x) < 0$ on $(0,n-m)$, the function
$f(x)$ is log-concave and hence unimodal.

To identify the critical points of $f(x)$, it suffices to
determine where $g'(x) = 0$ since $\log x$ is increasing.
Using the previous expression for $g'(x)$, it follows that
\begin{align}
\label{eqn:first_derivative}
  g'(x) &= \log\bracks*{\frac{(n-x)^2 - m^2}{\mu^2 x^2}}.
\end{align}
Therefore, the critical points are the solutions of
$(n-x)^2 - m^2 = \mu^2 x^2$,
so we have
\begin{align*}
  x^* = \begin{cases}
    \frac{n^2 - m^2}{2n} & \text{if $\mu = 1$,}\\
    \frac{n - \sqrt{n^2 - (1-\mu^2)(n^2 - m^2)}}{1-\mu^2} & \text{otherwise.}
  \end{cases}
\end{align*}
It remains and suffices to show that $x^*$ is a local maximum since
$f(x)$ is unimodal.
Observing that
\begin{align*}
  \frac{\partial}{\partial x} \log f(x) = g'(x)
\end{align*}
and differentiating $f(x)=\exp(\log f(x))$ using the chain rule, the definition
of $x^*$ gives
\begin{align*}
  f''\parens*{x^*} &= e^{\log f\parens*{x^*}}
    \bracks*{g''\parens*{x^*} + 
      g'\parens*{x^*}^2}\\
   &= f\parens*{x^*} g''\parens*{x^*}.
\end{align*}
We know $f(x^*)>0$, so 
$f''(x^*)$ has the same sign as $g''(x^*) < 0$.
Therefore, $x^*$ is a local maximum of~$f(x)$.
Using the continuity of $f(x)$ on $[0,n-m]$ and log-concavity, 
$f(x^*)$ is a global maximum.
\end{proof}

\begin{remark}
It is worth noting that for $m = o(n)$, the asymptotic behavior
of the critical point is continuous as a function of $\mu > 0$.
In particular, it follows from \Cref{lem:critical_point} that $x^* \sim n/(1+\mu)$.
\end{remark}

\subsection{Asymptotic Behavior of the Maximum Log Marginal}

Now that we have a formula for $x^*$, and hence an expression for $f(x^*)$,
we want to show that
\begin{align*}
  \parens*{\frac{\mu}{1+\mu}}^{2n} f\parens*{x^*} \le e^{-n^c},
\end{align*}
for some constant $c > 0$.
Because there are polynomially-many marginals in the sum, this leads to an
exponentially small upper bound for $\Prob{S_{2n} = 2m}$.
Define the \emph{maximum log marginal} to be
\begin{align}
\label{eqn:h_def}
  h(n) \DEF -\log\bracks*{\parens*{\frac{\mu}{1+\mu}}^{2n} f\parens*{x^*}}.
\end{align}
Equivalently, we show that $h(n) \ge n^c$ for sufficiently large $n$
using asymptotic equivalences.

\begin{lemma}
\label{lem:symmetric_h}
The maximum log marginal $h(n)$ can be symmetrically expressed as
\begin{align*}
  h(n) = \parens*{n+m}\log\bracks*{
    \parens*{\frac{1+\mu}{\mu}} \parens*{1 - \frac{x^*}{n+m}}}
  + \parens*{n-m}\log\bracks*{
    \parens*{\frac{1+\mu}{\mu}} \parens*{1 - \frac{x^*}{n-m}}}.
\end{align*}
\end{lemma}

\begin{proof}
Grouping the terms of $h(n)$ by factors of $n$, $m$ and $x^*$ gives
\begin{align*}
  &n \log\bracks*{\parens*{\frac{1+\mu}{\mu}}^2
      \frac{\parens*{n-x^*}^2 - m^2}{(n+m)(n-m)}
    }
  + m \log\bracks*{\frac{(n-m)\parens*{n+m-x^*}}{(n+m)\parens*{n-m-x^*}}}
  + x^* \log\bracks*{\frac{\parens*{\mu x^*}^2}{\parens*{n-x^*}^2-m^2}}.
\end{align*}
Using \Cref{eqn:first_derivative}, observe that the last term is
\begin{align*}
  x^* \log\bracks*{\frac{\parens*{\mu x^*}^2}{\parens*{n-x^*}^2-m^2}}
  = -x^* g'\parens*{x^*}
  = 0.
\end{align*}
The proof follows by
grouping the terms of the desired expression by factors of $n$ and $m$.
\end{proof}

The following lemma is the crux of our argument, as it presents an asymptotic
equality for the maximum log marginal in the PMF for correlated random walks.
We remark that we attempted to bound this quantity directly using Taylor
expansions instead of an asymptotic equivalence, and while this seems possible,
the expressions are unruly. Our asymptotic equivalence demonstrates
that second derivative information is needed, which makes
the earlier approach even more unmanageable.

\begin{lemma}
\label{lem:asymptotic_equality}
For any $\mu > 0$ and $m = o(n)$, the maximum log marginal satisfies
$h(n) \sim m^2/(\mu n)$.
\end{lemma}

\begin{proof}
The proof is by case analysis for $\mu$. In both cases we 
analyze $h(n)$ as expressed in \Cref{lem:symmetric_h}, consider
a change of variables, and use L'Hospital's rule twice.
In the first case, we assume $\mu = 1$.
The value of $x^*$ in \Cref{lem:critical_point} gives us
\begin{align*}
  1 - \frac{x^*}{n+m} &= \frac{2n(n+m) - \parens*{n^2 - m^2}}{2n(n+m)} = \frac{n+m}{2n}\\
  1 - \frac{x^*}{n-m} &= \frac{2n(n-m) - \parens*{n^2 - m^2}}{2n(n-m)} = \frac{n-m}{2n}.
\end{align*}
It follows that $h(n)$ can be simplified as
\begin{align*}
  h(n) &= n\log\bracks*{\parens*{\frac{1+\mu}{\mu}}^2 \parens*{\frac{n^2-m^2}{4n^2}}}
   + m\log\parens*{\frac{n+m}{n-m}}\\
   &= n\log\parens*{1 - \frac{m^2}{n^2}} + m \log\parens*{1 + \frac{2m}{n-m}}.
\end{align*}

To show $h(n)\sim m^2/n$, by the definition of asymptotic equivalence
we need to prove that
\begin{align*}
  \lim_{n\rightarrow\infty} \frac{n\log\parens*{1 - \frac{m^2}{n^2}} + m \log\parens*{1 + \frac{2m}{n-m}}}{\frac{m^2}{n}} = 1.
\end{align*}
Make the change of variables $y=m/n$. Since $m = o(n)$, this is equivalent to showing 
\begin{align*}
  \lim_{y \rightarrow 0} \frac{\log\parens*{1 - y^2} + y \log\parens*{1 + \frac{2y}{1-y}}}{y^2} = 1.
\end{align*}
Using L'Hospital's rule twice with the derivatives
\begin{align*}
  \frac{\partial}{\partial y} \bracks*{\log\parens*{1 - y^2} + y \log\parens*{1 + \frac{2y}{1-y}}} &=
  \log\parens*{-\frac{y+1}{y-1}}\\
  \frac{\partial^2}{\partial y^2} \bracks*{\log\parens*{1 - y^2} + y \log\parens*{1 + \frac{2y}{1-y}}} &=
  \frac{2}{1-y^2},
\end{align*}
it follows that
\begin{align*}
  \lim_{y \rightarrow 0} \frac{\log\parens*{1 - y^2} + y \log\parens*{1 + \frac{2y}{1-y}}}{y^2} 
  &=
  \lim_{y \rightarrow 0} \frac{\log\parens*{-\frac{y+1}{y-1}}}{2y} 
  = 
  \lim_{y \rightarrow 0} \frac{\frac{2}{1-y^2}}{2}
  = 1.
\end{align*}
This completes the proof for $\mu = 1$.

The case when $\mu \ne 1$ is analogous but messier.
Making the same change of variables $y=m/n$, it is equivalent to show that
\begin{align}
\label{eqn:biased_goal}
  &(1+y)\log\bracks*{\parens*{\frac{1+\mu}{\mu}}
  \parens*{1 - \frac{1}{1-\mu^2}\cdot\frac{1}{1+y}
  \cdot \parens*{1 - \sqrt{\mu^2 + \parens*{1-\mu^2}y^2}}}
  } \nonumber \\
  &\hspace{1.25cm}+ (1-y)\log\bracks*{\parens*{\frac{1+\mu}{\mu}}
  \parens*{1 - \frac{1}{1-\mu^2}\cdot\frac{1}{1-y}\cdot
  \parens*{1 - \sqrt{\mu^2 + \parens*{1-\mu^2}y^2}}}
  }
  \sim \mu^{-1} y^2,
\end{align}
because the value of $x^*$ for $\mu \ne 1$ in \Cref{lem:critical_point} gives us
\begin{align*}
  1 - \frac{x^*}{n+m}
  &=
  1 - \frac{1}{n+m}\cdot\frac{n}{1-\mu^2}\cdot\parens*{1 - \sqrt{\mu^2+\parens*{1-\mu^2}\frac{m^2}{n^2}}}.
\end{align*}
Denoting the left-hand side of \Cref{eqn:biased_goal} by $g(y)$,
one can verify that the first two derivatives of $g(y)$ are
\begin{align*}
  g'(y) &= \log\parens*{\frac{\mu^2 - \sqrt{\mu^2-\mu^2y^2 + y^2} + \parens*{\mu^2-1}y}{(\mu-1)\mu(y+1)}}
    - \log\parens*{\frac{-\mu^2 + \sqrt{\mu^2-\mu^2y^2 + y^2} + \parens*{\mu^2-1}y}{(\mu-1)\mu(y-1)}}\\
  g''(y) &= \frac{2}{(1+y)(1-y)\sqrt{y^2 - \mu^2\parens*{y^2 - 1}}}.
\end{align*}
Observing that $g(0) = g'(0) = 0$ due to convenient cancellations
and using L'Hospital's rule twice,
\begin{align*}
  \lim_{y\rightarrow 0} \frac{g(y)}{\mu^{-1} y^2}
  = 
  \lim_{y\rightarrow 0} \frac{g'(y)}{2\mu^{-1} y}
  =
  \lim_{y\rightarrow 0}  \frac{2}{(1+y)(1-y)\sqrt{y^2 - \mu^2\parens*{y^2-1}}} \cdot \frac{\mu}{2}
  = 1.
\end{align*}
This completes the proof for all cases of $\mu$.
\end{proof}

\CorrelatedTailBound*

\begin{proof}
For $n$ sufficiently large, 
the asymptotic equality for $h(n)$ in \Cref{lem:asymptotic_equality} gives us
\[
  h(n) \ge \parens*{1-\frac{\varepsilon}{2}} \frac{m^2}{\mu n}.
\]
It follows from our construction of $f(x)$ and the definition of the
maximum log marginal that
\begin{align*}
  \Prob{S_{2n} = 2m} &\le \poly(n) \cdot \parens*{\frac{\mu}{1+\mu}}^{2n} f\parens*{x^*}\\
  &= \poly(n) \cdot e^{-h(n)}\\
  &\le \poly(n) \cdot e^{-\parens*{1-\frac{\varepsilon}{2}} \frac{m^2}{\mu n}}\\
  &\le e^{-\parens*{1-\varepsilon}\frac{m^2}{\mu n}},
\end{align*}
as desired.
\end{proof}

\section{Conclusion}
\label{sec:conclusion}

We have made significant progress towards rigorously establishing the
conjectured slow regions of the phase diagram for the six-vertex model.
In particular, we prove that there exist boundary conditions for which Glauber dynamics
requires exponential convergence time for the entire ferroelectric
region and most of the antiferroelectric region.  Furthermore, our proofs
demonstrate why sharp boundaries exist between the ferroelectric phase and the
disordered phase, where Glauber dynamics is believed to transition to
polynomial-time convergence.
We have not fully characterized the antiferroelectric phase,
but our improvement over the best previous bounds in~\cite{liu}
cover a significantly larger part of the region.

Our arguments for the slow mixing of Glauber dynamics completely break down in
the disordered phase, as expected, but there has not been any rigorous work
showing that in this region of the phase diagram we have fast convergence.
The single exception is the unweighted case when we have $a=b=c$, which corresponds to
Eulerian orientations of the lattice region. This was shown to converge in
polynomial time for all boundary conditions~\cite{rt, luby2001markov, goldberg2004random}.
The approaches in these works are inherently combinatorial, and it seems that
generalizing them to weighted cases will require significantly different ideas.
Lastly, we emphasize that our proofs of slow mixing rely on
new techniques for analyzing lattice models,
which include
the closed-form
generating function for weighted non-backtracking walks
derived in~\Cref{sec:antiferroelectric} and
the exponentially small tail inequality for correlated random
walks developed in~\Cref{sec:tail-behavior}.

\bibliographystyle{alpha}
\bibliography{references}

\end{document}